\UseRawInputEncoding
\documentclass[12pt]{article} \pdfoutput=1%For ARXIV     
\usepackage[top=1in, bottom=1in, left=0.8in, right=0.8in]{geometry}
\usepackage{graphicx}  \usepackage{tabularx}
\graphicspath{essay/}
\usepackage{authblk}
\usepackage{pdflscape}
\usepackage{amssymb}
\usepackage{amsmath}
\usepackage{amsthm}
\usepackage{cases}
\usepackage{multirow} \usepackage{booktabs} \usepackage{hhline}
\usepackage{array} \usepackage{longtable}
\usepackage[ruled,vlined]{algorithm2e}

\usepackage{listings}
\numberwithin{equation}{section} 
\usepackage{amscd,epstopdf}
\usepackage{graphicx}\usepackage{multirow}
\usepackage{float}
\usepackage{subcaption}
\usepackage{comment}

\usepackage{mathrsfs}  %花体字母
\usepackage{caption}
\usepackage{afterpage} \usepackage{rotating}
\usepackage{enumerate}
\usepackage{stfloats}
\usepackage{tikz}
\usetikzlibrary{positioning} %为了实现相对位置的设定
\usepackage{xcolor} %为了实现不同的颜色
\usepackage[colorlinks,
linkcolor=blue,
anchorcolor=blue,
citecolor=blue]{hyperref}%加入文献链接
\newcommand{\keywords}[1]{\par\noindent\textbf{Keywords:} #1}

\newtheorem{thm}{Theorem}[section]
\newtheorem{proposition}[thm]{Proposition}

\newtheorem{definition}[thm]{Definition}
\usepackage{setspace}
\usepackage{indentfirst} 
\usepackage{graphicx}
\newtheorem{theorem}{Theorem}[section]
\newtheorem{lemma}[theorem]{Lemma}
\newtheorem{example}{Example}[section]
\renewcommand{\and}{\quad \quad}
	\title{Approximating Quantum States with Positive Partial Transposes in Multipartite System via Linearized Proximal Alternative Direction Method of Multipliers}
	\author{Jingwen Fan\textsuperscript{a}	\and 
	Deren Han\textsuperscript{a} \and Lin Chen\textsuperscript{a}\thanks{Corresponding Author: linchen@buaa.edu.cn}}	
	\affil{\textsuperscript{a} LMIB, Ministry of education, and School of Mathematical Sciences, Beihang University, Beijing 100191, China}
	\date{}
	
\begin{document}
	\maketitle
	\begin{abstract}
		 Numerical approximation of quantum states via convex combinations of states with positive partial transposes (bi-PPT state) in multipartite systems constitutes a fundamental challenge in quantum information science. We reformulate this problem as a linearly constrained optimization problem. An approximate model is constructed through an auxiliary variable and a suitable penalty parameter, balancing constraint violation and approximation error. To slove the approximate model, we design a Linearized Proximal Alternating Direction Method of Multipliers (LPADMM), proving its convergence under a prescribed inequality condition on regularization parameters. The algorithm achieves an iteration complexity of $O(1/\epsilon^2)$ for attaining $\epsilon$-stationary solutions. Numerical validation on diverse quantum systems, including three-qubit W/GHZ states and five-partite GHZ and multiGHZ states with noises, confirms high-quality bi-PPT approximations and decomposability certification, demonstrating the utility of our method for quantum information applications.
	\end{abstract}

	
\section{Introduction}
%量子纠缠的基础理论与应用价值量子纠缠作为量子计算核心资源（引用Einstein-Podolsky-Rosen原始论文）多体纠缠在量子通信/计算中的关键作用（引用Horodecki综述）PPT判据的理论地位（Peres-Horodecki准则原始文献）
Quantum entanglement, first articulated by Einstein, Podolsky, and Rosen \cite{einstein1935can}, has emerged as a cornerstone of modern quantum technologies, driving breakthroughs in computing \cite{shor1994algorithms}, secure communication \cite{ekert1991quantum}, and simulation \cite{georgescu2014quantum}. A central challenge lies in characterizing quantum states, particularly those expressible as convex combinations of bi-PPT states. The study of bi-PPT decompositions originates from the Peres-Horodecki criterion \cite{peres1996separability,horodecki1996necessary}, which established positivity under partial transposition as a necessary condition for separability. While being both necessary and sufficient for bipartite systems of dimensions $2\times2$ and $2\times3$, the situation becomes significantly more complex in multipartite scenarios. Previous theorectical work contains the results for low-demensional system and the general multipartite system. Biseparability was analyzed in a five-dimensional set of states within three-subsystem quantum systems of arbitrary finite dimension \cite{eggeling2001separability}, and some specific examples for the biseparable states in several systems were also constructed \cite{han2019construction}. Recent research has further elucidated the properties of bi-PPT states in general settings. Some analytic examples were constructed to disprove that biseparability is equivalent to being a PPT mixture for general cases \cite{ha2016construction}, and the conjecture that the completely symmetric states are separable if and only if it is a convex combination of symmetric pure product states was proven to be true for both bipartite and multipartite cases \cite{chen2019separability}. Recent work highlighted fundamental differences between bipartite and multipartite biseparability \cite{choudhary2025unconditionally}. Despite substantial theoretical progress for bi-separbility, determining if a state can be decomposed as a convex combination of bi-PPT states is a more complex problem. Numerical methods for determining bi-PPT decomposability in $n$-partite systems remain underexplored. The exponential growth of Hilbert space dimension makes conventional separability criteria \cite{horodecki2009quantum} computationally prohibitive. Moreover, deciding bi-PPT decomposability is strongly NP-hard for $n\ge4$ \cite{gharibian2008strong}. This fundamental limitation necessitates novel optimization paradigms. Prior efforts in entangled state approximation have predominantly focused on bipartite systems, where techniques such as convex optimization \cite{han2013successive}, semidefinite relaxations \cite{Lancien_2015,li2020separability}, and projected gradient flows \cite{lin2022low} have achieved notable success. Recent efforts have also explored machine learning techniques \cite{chen2021entanglement} and neural network representations \cite{asif2023entanglement} for entanglement detection, yet scalable and general methods for bi-PPT approximation are still lacking.

%The theoretical significance of biseparable lies in its role as a critical relaxation of full separability, enabling practical entanglement detection where exact separability tests are computationally infeasible.%These serve as essential benchmarks for separability and entanglement detection owing to the necessary condition imposed by the Peres-Horodecki criterion \cite{peres1996separability}. States that cannot be decomposed in this manner are necessarily entangled, allowing for efficient resource identification without full state tomography. Beyond foundational insights, such decompositions underpin device-independent quantum security protocols, exemplified by Ekert’s quantum key distribution framework, which leverages entanglement to guarantee cryptographic security. In multipartite systems, tests based on bi-PPT decompositions can reveal bound entanglement—entangled states that are not distillable yet play critical roles in quantum error correction \cite{wang2023universal}. Thus, the study of bi-PPT decompositions lies at the heart of both quantum foundations and scalable technologies, bridging theoretical inquiry with practical applications in resource theories and hardware-efficient algorithms.%, with empirical results \cite{ioannou2022simulability} confirming exponential time complexity beyond 8 qubits %A complete family of separability criteria based on semidefinite programming has been proposed in \cite{doherty2004complete}. 

%理论结果

%ADMM的理论突破与量子适配性
To address this gap, we reformulate the problem into a tractable optimization model solvable via an adapted Alternating Direction Method of Multipliers (ADMM), enabling the identification of stationary points. ADMM has proven highly effective in diverse applications including image processing \cite{yang2015alternating}, mobile edge intelligence \cite{he2024admm}, matrix completion and separation \cite{shen2014augmented,sun2014alternating,xu2012alternating,gao2024low} and neuroscience\cite{sun2016deep}. Theoretical advancements leverage the Kurdyka-Łojasiewicz (K\L) inequality for convergence guarantees \cite{attouch2010proximal}, which provides a powerful framework for analyzing the convergence of nonconvex optimization algorithms by ensuring that the objective function and linear mapping satisfy certain properties. This has been complemented by subsequent developments, including globally convergent variants with quadratic regularization \cite{deng2016global,hong2016convergence} that enhance stability and avoid divergent behavior in ill-conditioned problems. Extensions to nonsmooth composite objectives \cite{guo2017convergence} have broadened the scope of ADMM to handle functions involving norms and other non-differentiable terms, which are prevalent in sparse optimization and compressed sensing. Moreover, relaxed continuity conditions \cite{wang2019global} have been introduced to broaden the applicability of ADMM to problems where traditional smoothness assumptions are violated. These advances make ADMM a robust and versatile framework for tackling the challenging problem of bi-PPT decomposition.%Mechanistic reinterpretations \cite{eckstein2015understanding} further enriched its theoretical foundations. The linear and sublinear convergence rate of multi-block ADMM have also been established \cite{lin2015sublinear,davis2017faster}. Recent innovations include combining ADMM with Riemannian optimization to handle manifold constraints efficiently \cite{li2022riemannian}, adaptive penalty parameter schemes under semi-strong convexity assumptions to accelerate convergence \cite{tang2024self}. 

In this work, we develop a comprehensive optimization framework for approximating multipartite quantum states with the from in (\ref{proform}), advancing beyond conventional witness-based semi-definite programming detection methods \cite{jungnitsch2011taming}. Our method employs linear transformations of partially transposed matrices and variable substitution to establish a novel mathematical model that inherently preserves quantum correlations while imposing linear constraints. To overcome the poor properties of indicator functions, we introduce a penalty term to balance objective minimization and constraint violation. The resulting non-separable objective function, characterized by its marginal convexity and smoothness, is addressed via the LPADMM. This algorithm decomposes the problem into tractable subproblems including coefficient optimization under simplex constraints (enforced via MATLAB's `quadprog'), positive semidefiniteness enforcement through spectral projections, and trace normalization via KKT-conditions. We establish theoretical guarantees including monotonicity of the augmented Lagrangian, bounded subgradients under Lipschitz assumptions, and global convergence to stationary points using the K{\L} property. An iteration complexity of $O(1/\epsilon^2)$ for $\epsilon$-stationary solutions is also provided. Numerical validation employs four distinct test cases: W and GHZ states with white noises in a three-qubit system, and GHZ and multiGHZ states with noises in a five-partite system. Although the problem is inherently nonconvex, conducting multiple experimental runs and selecting the solution with the smallest error enables effective approximation via convex combinations of bi-PPT states. Numerical experiments demonstrate that our approach achieves a satisfactory level of approximation.

This paper is organized as follows. Section \ref{Notation} introduces foundational notations and optimization prerequisites. The problem formulation and LPADMM framework are then presented in Section \ref{prodes}. Numerical experiments are provided in Section \ref{numexp}. Finally, Section \ref{con} concludes the paper. The convergence proof and complexity analysis are given in the Appendix.
\section{Notation and Preliminaries}\label{Notation}
%介绍问题：
\subsection{Convex Combination of bi-PPT States}
To formalize the problem, we clarify with the three-partite case $ABC:=\mathcal{H}^{n_1}\otimes\mathcal{H}^{n_2}\otimes\mathcal{H}^{n_3}$, noting that extension to $n$-partite systems follows analogous principles. Within this Hilbert space, mixed states are represented by density matrices $\rho_{n_1n_2n_3\times n_1n_2n_3}$ (positive semi-definite with unit trace), which exhibit critical structural properties. A three-partite state admitting decomposition as a convex combination of bi-PPT states satisfies
\begin{align}\label{proform}
	\rho=a\alpha+b\beta+c\gamma ,\quad a+b+c=1, a, b, c\ge0,
\end{align}
where $\alpha, \beta, \gamma$ satisfy $\alpha\succeq0$, $\alpha^{\Gamma_A}\succeq0$, $\beta\succeq0$, $\beta^{\Gamma_B}\succeq0$, $\gamma\succeq0$, $\gamma^{\Gamma_C}\succeq0$ with unit trace ($\mathrm{tr}\left( \alpha\right) =\mathrm{tr}\left( \beta\right) =\mathrm{tr}\left( \gamma\right) =1$). Here $\alpha^{\Gamma_A}$, $\beta^{\Gamma_B}$ and $\gamma^{\Gamma_C}$ denote the partial transpositions over subsystems $A$, $B$ and $C$ respectively. Partial transpositions are defined via block-wise permutation as follows:
\begin{itemize}
	\item Subsystem $A$: For $\alpha$ partitioned into $n_1\times n_1$ blocks $F_{i,j}\in\mathbb{R}^{n_2n_3\times n_2n_3}$, which means 
	$\alpha=\begin{pmatrix}F_{1,1}&\cdots&F_{1,n_1}\\
		\vdots&\ddots&\vdots\\
		F_{n_1,1}&\cdots&F_{n_1,n_1}
	\end{pmatrix}$, we have $$\alpha^{\Gamma_A}=\begin{pmatrix}F_{1,1}&\cdots&F_{n_1,1}\\
		\vdots&\ddots&\vdots\\
		F_{1,n_1}&\cdots&F_{n_1,n_1}
	\end{pmatrix}.$$
	\item Subsystem $B$: For $\beta$ partitioned into $n_1\times n_1$ blocks $G_{i,j}$, where each $G_{i,j}$ is further partitioned into $n_2\times n_2$ blocks  
	$H_{i,j}^{p,q}\in\mathbb{R}^{n_3\times n_3}$, we have
	$$\beta^{\Gamma_B}=\begin{pmatrix}G_{1,1}&\cdots&G_{1,n_1}\\
		\vdots&\ddots&\vdots\\
		G_{n_1,1}&\cdots&G_{n_1,n_1}
	\end{pmatrix},\quad G_{i,j}=\begin{pmatrix}H^{1,1}_{i,j}&\cdots&H^{n_2,1}_{i,j}\\
		\vdots&\ddots&\vdots\\
		H^{1,n_2}_{i,j}&\cdots&H^{n_2,n_2}_{i,j}
	\end{pmatrix}.$$
	This definition implies that the partial transposition over subsystem $B$ permutes only the innnermost block indices while preserving the outer block stucture.
	
	\item Subsystem $C$: Implemented via element-wise matrix transpose on each $n_3\times n_3$ block ($H_{i,j}^{p,q}$). 
\end{itemize}

For an $n$-partite system, bi-PPT states require positivity under partial transpositions of at least two subsystems. The composite operation is achieved through sequential single-subsystem transpositions. For example, we consider a four-partite system $ABCD:=\mathcal{H}^{n_1}\otimes\mathcal{H}^{n_2}\otimes\mathcal{H}^{n_3}\otimes\mathcal{H}^{n_4}$ partitioned as $AB\otimes CD$. The bi-PPT is enforced by
$$\rho^{\Gamma_{AB}}=\left( \rho^{\Gamma_{A}}\right) ^{\Gamma_{B}},$$
where $\Gamma_{A}$ and $\Gamma_{B}$ act respectively on the first and second subsystems of the $AB$ partition.

\subsection{Optimization Preliminaries}
In this section, we introduce some basic notations, optimization definitions and lemmas for this paper. The field of real numbers is denoted by $\mathbb{R}$ and the identity matrix with size $m\times m$ is denoted by $I_{m}$. Throughout this paper, $\left( \cdot\right) ^\top$ means the transpose of a vector or matrix. $\mathrm{mat}\left( \cdot\right) $ transforms a vector into a matrix and $\mathrm{vec}\left( \cdot\right) $ transforms a matrix into a vector. The notation $\left\|\cdot\right\|$ denotes the Euclidean norm for vectors and the Frobenius norm for matrices. Let $M$ be a positive definite matrix; the $M$-norm of a vector $x$ is defined as $\left\|x\right\|_M:=\sqrt{\left\langle x,Mx\right\rangle }$. For symmetric matrices $M, N\in\mathbb{R}^{n\times n}$, $M\succeq N$ means that $M-N$ is positive semi-definite. The following are some preliminaries on optimization theory.
\begin{definition}
	A function $f$ defined on set $Q$ is said to be L-smooth if it is differentiable and there exists a constant $L>0$ such that for all $x,y\in Q$, the following inequality holds:
	\begin{align*}
		\left\|\nabla f\left( x\right) -\nabla f\left( y\right) \right\|\leqslant L\left\|x-y\right\|.
	\end{align*}
\end{definition}
\begin{lemma}\label{convex_twice}
\cite{nesterov2018lectures} A twice continuously defferentiable function $f$ defined on an open convex set $Q$ is convex if and only if for any $x\in Q$, we have
\begin{align*}
	\nabla^2f\left( x\right) \succeq0.
\end{align*}
\end{lemma}
\begin{lemma}
	\cite{nesterov2018lectures} If $f$ is L-smooth and twice continuously differentiable on set Q, then the following inequalities are equivalent.
	\begin{align*}
		&f\left( y\right) -f\left( x\right) -\left\langle \nabla f\left( x\right) ,y-x\right\rangle\leqslant\frac{L}{2}\left\|x-y\right\|^2,\\
		&\nabla^2f\left( x\right) \preceq LI.
	\end{align*}
\end{lemma}

\begin{lemma}\cite{nocedal1999numerical} [KKT conditions]Consider the minimization problem
\begin{align*}
\min_{x\in\mathbb{R}^{n}}\quad&f\left( x\right) \\
\mathrm{s.t.}\quad& 
\begin{cases}
	c_i\left( x\right) =0, \quad i\in\mathcal{E},\\
	c_i\left( x\right) \ge0,\quad i\in\mathcal{I},
\end{cases}
\end{align*} 
where $f$ and the functions $c_i$ are all smooth, real-valued functions on a subset of $\mathbb{R}^n$, and $\mathcal{E}$ and $\mathcal{I}$ are two finite sets of indices. $c_i, i\in\mathcal{E}$ are the equality constraints and $c_i, i\in\mathcal{I}$ are the inequality constraints.

Suppose that $x^*$ is a local solution of the problem, that the functions are continuously differentiable, and that the set of active constraint gradients is linearly independent (LICQ holds) at $x^*$. Then there is a Larange multiplier vector $\lambda^*$, with components $\lambda_i^*, i\in\mathcal{E}\cup\mathcal{I}$, such that the following conditions are satisfied at $\left( x^*,\lambda^*\right) $
\begin{align*}
	&\nabla_x\left( f\left( x\right) -\sum_{i\in\mathcal{E}\cup\mathcal{I}}\lambda_ic_i\left( x\right) \right) \left( x^*,\lambda^*\right) =0,\\
	&c_i\left( x^*\right) =0,\quad\forall i\in\mathcal{E},\\
	&c_i\left( x^*\right) \ge0,\quad\forall i\in\mathcal{I},\\
	&\lambda_i^*\ge0,\quad\forall i\in\mathcal{I},\\
	&\lambda_i^*c_i\left( x^*\right) =0,\quad\forall i\in\mathcal{E}\cup\mathcal{I}.
\end{align*}
\end{lemma}

\section{The Approximation Method}\label{prodes}
This section formulates the quantum state decomposition task as a linearly constrained optimization problem. Using the three-partite system as a paradigm, the theoretical analysis naturally extends to $n$-partite systems. Let $\mathcal{A}, \mathcal{B}, \mathcal{C}$ denote linear operators acting on symmetric matrices $\alpha$, $\beta$, and $\gamma$ through the transformations $\mathcal{A}\left( \alpha\right) =\alpha^{\Gamma_A}$, $\mathcal{B}\left( \beta\right) =\beta^{\Gamma_B}$, $\mathcal{C}\left( \gamma\right) =\gamma^{\Gamma_C}$, respectively. We introduce auxiliary variables $\phi=\mathcal{A}\left( \alpha\right) $, $\psi=\mathcal{B}\left( \beta\right) $, $\omega=\mathcal{C}\left( \gamma\right) $, along with trace-normalized matrices $L$, $S$ and $T$. With a given state $\rho$, the optimization problem is formulated as follows
\begin{align*}
	\min\quad&f\left( a,b,c,\alpha,\beta,\gamma\right) =\frac{1}{2}\left\|\rho-a\alpha-b\beta-c\gamma \right\|^2\\\notag
	s.t.\quad&\mathcal{A}\left( \alpha\right) =\phi,\mathcal{B}\left( \beta\right) =\psi,\mathcal{C}\left( \gamma\right) =\omega,\\\notag
	&\alpha-L=0,\beta-S=0,\gamma-T=0,\\\notag
	&\alpha\succeq0, \phi\succeq0, \mathrm{tr}\left( L\right) =1,\\\notag
	&\beta\succeq0, \psi\succeq0, \mathrm{tr}\left( S\right) =1,\\\notag
	&\gamma\succeq0, \omega\succeq0, \mathrm{tr}\left( T\right) =1,\\\notag
	&a+b+c=1, a\ge0, b\ge0, c\ge0.\notag
\end{align*}

To facilitate analysis, we vectorize all matrix variables with the following notation
\begin{align*}
	y&=\left( a,b,c\right) ^\top,\\
	x&=\left( \mathrm{vec}\left( \alpha\right)^\top ,\mathrm{vec}\left( \beta\right)^\top ,\mathrm{vec}\left( \gamma\right)^\top \right) ^\top,\\
	z&=\left(\mathrm{vec}\left( \phi\right)^\top ,\mathrm{vec}\left( \psi\right)^\top ,\mathrm{vec}\left( \omega\right)^\top ,\mathrm{vec}\left( L\right)^\top ,\mathrm{vec}\left( S\right)^\top ,\mathrm{vec}\left( T\right)^\top \right) ^\top.
\end{align*}
The linear operator $\mathcal{A}: \mathbb{R}^{n^2}\to \mathbb{R}^{n^2}$ admits a matrix representation $A_{\alpha}$ with block structure
\begin{align*}
	A_{\alpha}\left( id_1,id_2\right) =I_{n_2n_3},
\end{align*}
where index sets $id_1$ and $id_2$ are defined as
\begin{align*}
	id_1 &= \left\lbrace \left( l-1\right) n_2n_3 + 1 : ln_2n_3 \mid l = 1,\ldots,n^2\right\rbrace , \\
	id_2 &= \left\lbrace \left( k-1\right) n + \left( j-1\right) n_1 + i \mid i,k=1,\ldots,n_1;\ j=1,\ldots,n_2n_3\right\rbrace 
\end{align*}
with $n=n_1n_2n_3$. Operators $\mathcal{B}$ and $\mathcal{C}$ have analogous matrix representations $B_{\beta}$ and $C_{\gamma}$.

The problem is reformulated in standard constrained form:
\begin{align}\label{problem}
	\min&\quad f\left( x,y\right) +\delta_{\mathcal{Z}}\left( z\right) +\delta_{\mathcal{Y}}\left( y\right) +\delta_{\mathcal{X}}\left( x\right) \notag\\
	\mathrm{s.t.}&\quad Ax-z=0,
\end{align}
where
\begin{itemize}
	\item $A$ denotes the full column-rank block-diagonal operator:
	\begin{align*}
		A = \begin{pmatrix}
			A_\alpha & 0 & 0 \\
			0 & B_\beta & 0 \\
			0 & 0 & C_\gamma \\
			I_{n^2} & 0 & 0 \\
			0 & I_{n^2} & 0 \\
			0 & 0 & I_{n^2}
		\end{pmatrix}.
	\end{align*}
\item $f\left( x,y\right) =\frac{1}{2}\left\|\rho-a\alpha-b\beta-c\gamma \right\|^2$ is a smooth and marginally convex objective function (See Appendix \ref{app1}).
\item $\delta_{\mathcal{X}}$,$\delta_{\mathcal{Y}}$, $\delta_{\mathcal{Z}}$ are indicator functions for the convex sets:
\begin{align*}
	\mathcal{X} &= \left\lbrace x  \mid  \alpha  \succeq 0,\ \beta  \succeq 0,\  \gamma  \succeq 0 \right\rbrace , \\
	\mathcal{Y} &= \left\lbrace y \mid a,b,c \geq 0,\ a+b+c=1 \right\rbrace , \\
	\mathcal{Z} &= \left\lbrace  z  \mid  \phi\succeq 0 , \psi\succeq 0 ,\omega  \succeq 0,\ \text{tr}\left( L\right) =\text{tr}\left( S\right) =\text{tr}\left( T\right) =1 \right\rbrace .
\end{align*}
\end{itemize}

To address the limitations of the indicator function in the original problem (\ref{problem}), we propose an approximate model incorporating a quadratic penalty term $\frac{\xi}{2}\left\|p-z\right\|^2$, where the positive hyperparameter $\xi$ balances constraint adherence and objective minimization. When $\xi\to+\infty$, the problem (\ref{proapp}) is equivalent to problem (\ref{problem}).  This yields the regularized optimization problem:
\begin{align}\label{proapp}
	\min&\quad h\left( x,y,z,p\right) =f\left( x,y\right) +\delta_{\mathcal{Z}}\left( p\right) +\delta_{\mathcal{Y}}\left( y\right) +\delta_{\mathcal{X}}\left( x\right) +\frac{\xi}{2}\left\|p-z\right\|^2\notag\\
	s.t.&\quad Ax-z=0,
\end{align}
with the corresponding augmented Lagrangian
\begin{align}\label{aulag}
	\mathcal{L}_{\eta}\left( x,y,p,z,\lambda\right) =h\left( x,y,z,p\right) +\left\langle \lambda,Ax-z\right\rangle+\frac{\eta}{2}\left\|Ax-z\right\|^2,
\end{align}
where $\lambda$ is the Lagrange multiplier and $\eta>0$ is the penalty parameter. While the classical ADMM iteration (\ref{ADMM}) struggles with non-separability in the $x$-subproblem, we introduce a partial linearization technique. 

\begin{equation}\label{ADMM}
\left\{\begin{aligned}
	y^{k+1} &=\arg\min_{y\in\mathbb{R}^3}\mathcal{L}_{\eta}\left( x^k,y,p^k,z^k,\lambda^k\right) ,\\
	x^{k+1} &=\arg\min_{x\in\mathbb{R}^{n^2}}\mathcal{L}_{\eta}\left( x,y^{k+1},p^k,z^k,\lambda^k\right) ,\\
	p^{k+1} &=\arg\min_{p\in\mathbb{R}^{n^2}}\mathcal{L}_{\eta}\left( x^{k+1},y^{k+1},p,z^k,\lambda^k\right) ,\\
	z^{k+1} &=\arg\min_{z\in\mathbb{R}^{n^2}}\mathcal{L}_{\eta}\left( x^{k+1},y^{k+1},p^{k+1},z,\lambda^k\right) ,\\
	\lambda^{k+1}& =\lambda^k+\eta\left( Ax^{k+1}-z^{k+1}\right) .
\end{aligned}\right.
\end{equation}

Instead of linearizing the entire Lagrangian, we selectively linearize $f\left( x,y\right) $, exploiting its quadratic properties to derive the Linearized Proximal ADMM (LPADMM): 
\begin{numcases}{}
	y^{k+1} =\arg\min_{y\in\mathbb{R}^3}\left\lbrace \delta_{\mathcal{Y}}\left( y\right) +f\left( x^k,y\right) +\frac{\mu_1}{2}\left\|y-y^k\right\|^2\right\rbrace, \label{eq:case1} \\
	x^{k+1} =\arg\min_{x\in\mathbb{R}^{n^2}}\left\lbrace \delta_{\mathcal{X}}\left( x\right) +\left\langle \nabla_xf\left( x^k,y^{k+1}\right) ,x\right\rangle+\left\langle \lambda^k,Ax\right\rangle+\frac{\eta}{2}\left\|Ax-z^k\right\|^2+\frac{\mu_2}{2}\left\|x-x^k\right\|^2\right\rbrace, \label{eq:case2} \\
	p^{k+1} =\arg\min_{p\in\mathbb{R}^{n^2}}\left\lbrace \delta_{\mathcal{Z}}\left( p\right) +\frac{\xi}{2}\left\|p-z^k\right\|^2+\frac{\mu_3}{2}\left\|p-p^k\right\|^2\right\rbrace, \label{eq:case3} \\
	z^{k+1} =\arg\min_{z\in\mathbb{R}^{n^2}}\left\lbrace-\left\langle\lambda^k,z\right\rangle+\frac{\eta}{2}\left\|Ax^{k+1}-z\right\|^2+\frac{\xi}{2}\left\|p^{k+1}-z\right\|^2\right\rbrace, \label{eq:case4} \\
	\lambda^{k+1} =\lambda^k+\eta\left( Ax^{k+1}-z^{k+1}\right) , \label{eq:case5}
\end{numcases}
where $\mu_1$,$\mu_2$, $\mu_3$ are proximal regularization coefficients that stabilize convergence. This approach preserves the problem’s inherent structure while enabling tractable subproblem solutions. The method ensures feasibility via quadratic penalties and spectral projections, with theoretical guarantees of convergence established in Appendix \ref{convergence}.

\subsection{Subproblems}
The implementation of the LPADMM involves iterative solutions to three key subproblems. The $y$-subproblem constitutes a convex quadratic programming task, which we efficiently solve using MATLAB's built-in \texttt{quadprog} function. 

For the $x$-subproblem, the special structure $A^\top A=2I$ enables equivalence to a projection onto the positive semidefinite cone. This reduces to solving:
\begin{align*}
	\arg\min_{x\in\mathcal{X}}\;\left\|x-g\right\|_2^2,
\end{align*}
where $g=\frac{-A^\top\lambda^k+\mu_2x^k+\eta A^\top z^k-\nabla_xf\left( x^k,y^{k+1}\right) }{2\eta+\mu_2}$.

Considering the separability of the function, we restructure the vectorial optimization problem by decomposing and reshaping the vectors into matrices. This transformation leads to a new minimization problem below: $$\arg\min_{x\in\mathcal{X}}\left\|\begin{pmatrix}
	\alpha \\ \beta\\ \gamma
\end{pmatrix}-\begin{pmatrix}
	\mathrm{mat}\left( g_1\right) \\\mathrm{mat}\left( g_2\right) \\\mathrm{mat}\left( g_3\right) 
\end{pmatrix}\right\|_F^2.$$
The results can be obtained by projecting $g_1$, $g_2$, $g_3$ to the positive semi-definite subspace, which is resolved via spectral decomposition (Lemma \ref{proj}).
\begin{lemma}\label{proj}
	For any real symmetric matrix $A$ with eigenvalue decomposition $A=U\Sigma U^\top$, its projection onto the positive semidefinite cone is given by: 
	\begin{equation}
		\mathcal{P}_{\succeq 0}\left( A\right)  = U\Sigma_+ U^\top,\quad \text{where } \left( \Sigma_+\right) _{ii} = \max\left( \Sigma_{ii}, 0\right) .
	\end{equation}
\end{lemma}

The calculations of $\phi,\psi,\xi$ are similar. 

The computation of parameters $L, S, T$ requires solving trace-constrained problems of the form:
\begin{align*}
	\min &\quad\frac{1}{2}\left\|V-S\right\|_F^2\\
	s.t. &\quad \mathrm{tr}\left( S\right) =1.
\end{align*}
The Lagrangian formulation yields:
\begin{align*}
	\mathcal{L}=\frac{1}{2}\left\|V-S\right\|_F^2-\lambda\left( \mathrm{tr}\left( S\right) -1\right) .
\end{align*}
After calculating the partial gradient, we get
\begin{align}
	\frac{\partial\mathcal{L}}{\partial\lambda}=\mathrm{tr}\left( S\right) -1=0, \notag\\
	\frac{\partial\mathcal{L}}{\partial S}=S-V-\lambda I=0, \label{1}
\end{align}
yielding $S_{ii}-V_{ii}-\lambda=0, S_{\left( i,j\right) }=V_{\left( i,j\right) }$ for $i\ne j$.
According to $\mathrm{tr}\left( S\right) =\sum_{i=1}^nS_{ii}=1$, $\lambda$ can be obtained by $\sum_{i=1}^n\left( S_{ii}-V_{ii}\right) -n\lambda=0$. The results of $S$ can be got by substituting $\lambda$ into (\ref{1}).

Now we expand the above method to more complicated systems. For a four-partite system, we extend the formulation to:
\begin{align}
	\min\quad&f\left( p_1,\cdots, p_7,\alpha_1,\cdots,\alpha_7\right) =\frac{1}{2}\left\|\rho-p_1\alpha_1-\cdots-p_7\alpha_7\right\|^2\\
	s.t.\quad&\alpha_i\succeq0,tr\left( \alpha_i\right) =1,\notag\\
	&\alpha_1^{\Gamma_A}\succeq0,\alpha_2^{\Gamma_B}\succeq0,\alpha_3^{\Gamma_C}\succeq0,\alpha_4^{\Gamma_D}\succeq0,\notag\\
	&\alpha_5^{\Gamma_{AB}}\succeq0,\alpha_6^{\Gamma_{AC}}\succeq0,\alpha_7^{\Gamma_{AD}}\succeq0,\notag\\
	&p_1+\cdots+p_7=1, p_i\ge0,\notag\\
	&\forall i\in{1,2,\cdots,7}.\notag
\end{align}
The definitions are similiar to the three-partite system. The mapping from $\alpha_i$ to its transformed form is a linear mapping and the marginal convexity is kept. Also, we can introduce another fourteen auxiliary parameters to solve the problem.

For the $n$-partite system, when $n$ is odd, the number of the transformed matrices is $\begin{pmatrix}n\\1\end{pmatrix}+\begin{pmatrix}n\\2\end{pmatrix}+\cdots+\begin{pmatrix}n\\\frac{n-1}{2}\end{pmatrix}$. When $n$ is even, the number of the transformed matrices is $\begin{pmatrix}n\\1\end{pmatrix}+\begin{pmatrix}n\\2\end{pmatrix}+\cdots+\begin{pmatrix}n\\\frac{n}{2}\end{pmatrix}$. The sum equals to $2^{n-1}-1$, which means $2^n-2$ auxiliary parameters are needed.

\section{Numerical Experiments}\label{numexp}
To validate the effectiveness of our method for approximating quantum states through form (\ref{proform}), we conduct systematic numerical experiments across three scenarios: three-qubit W/GHZ states with noises and five-partite multiGHZ states with noises. The algorithm employs stopping criteria based on primal and dual residuals, configured with parameters $\mu_1=0.1, \mu_2=1.1, \mu_3=\xi, \eta=2\xi+1$ and convergence tolerance $\epsilon=10^{-8}$. Leveraging the orthogonal structure $A^\top A=2I$ and (\ref{prop:sufficient_descent}), we can tight the configuration by setting $\mu_2=0$ and $\eta>\max\left\lbrace L,2\xi\right\rbrace $. Each experiment executes 30 independent trials with randomized initial points to account for solution variability.

\begin{example}\label{ex1}
	Consider a three-qubit system with state construction: $\Omega=W_3W_3^\top+lI_8$, $W_3=e_0\otimes e_0\otimes e_1+e_0\otimes e_1\otimes e_0+e_1\otimes e_0\otimes e_0$ where $e_0=\begin{pmatrix}1\\0\end{pmatrix}, e_1=\begin{pmatrix}0\\1\end{pmatrix}$ and $l\in\mathbb{R}_{+}$. The normalized state becomes $\rho=\Omega/\mathrm{tr}(\Omega)$.
\end{example}
We have known that if $l>\sqrt{2}$, the state must be a bi-PPT state. We test $l=3$ and $l=0$ cases. For $l=3$, it must can be decomposed into the convex combination of bi-PPT states and the original objective value $f(x,y)$ can be arbitrarily small. To find a suitable penalty parameter, we sweep $\xi\in[100,1000]$ with 50-step increments, monitoring both $f(x,y)$ and constraint violation $\left\|p-z\right\|^2$. We also recorded whether the variables satisfied the constraints. Figure \ref{re1} demonstrates that when the state can be decomposed into the convex combination of bi-PPT states,$\xi=100$ suffices to achieve approximation precision $e^{-10}$ with constraint violation $e^{-16}$. This confirms our method's capability to find valid bi-PPT decompositions when theoretically permitted.
\begin{figure}[H]
	\centering
	\includegraphics[width=0.45\textwidth]{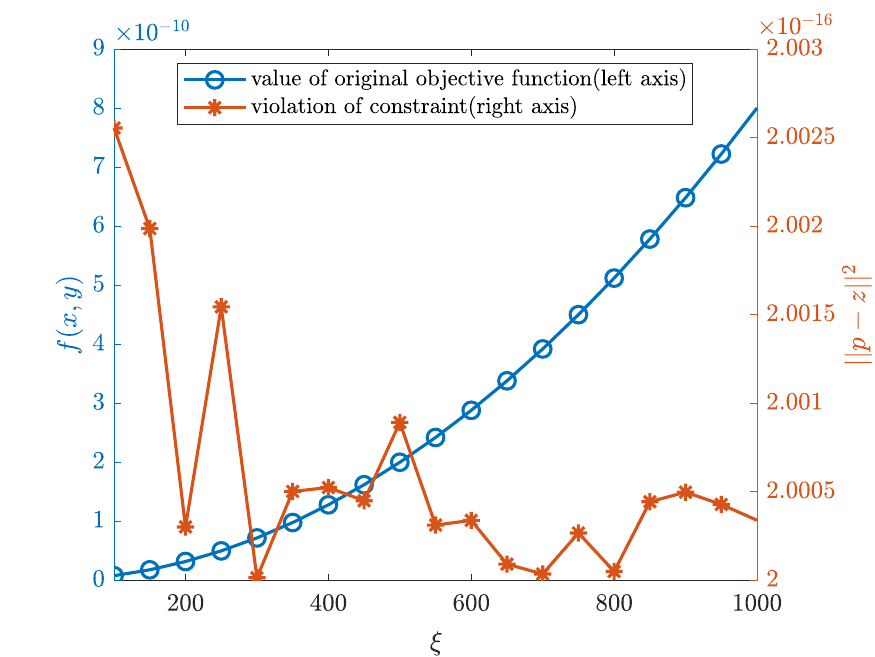}
	\caption{Original objective function $f(x,y)$ (left axis) and penalty term $\left\|p-z\right\|^2$ (right axis) under varying $\xi$ for state $W_3W_3^\top+3I_8$.}%设置标题，且会自动编号,图在下
	\label{re1}%设置标签
\end{figure}
Figure \ref{re2} reveals distinct behavior for $l=0$ (non-decomposable case). Here, $\xi$ requires escalation to $>800$ for violation control, with achievable precision limited to $e^{-1}$. The observed tradeoff between increasing $f(x,y)$ and decreasing $\left\|p-z\right\|^2$ confirms theoretical predictions. Despite residual $p-z$ discrepancies, all solutions satisfy problem (\ref{problem}) requirements. The $\xi$ plateau in $[800,1000]$ suggests it can be set as 900 to balance approximation error and constraint satisfaction.
\begin{figure}[H]
	\centering
	\includegraphics[width=0.4\textwidth]{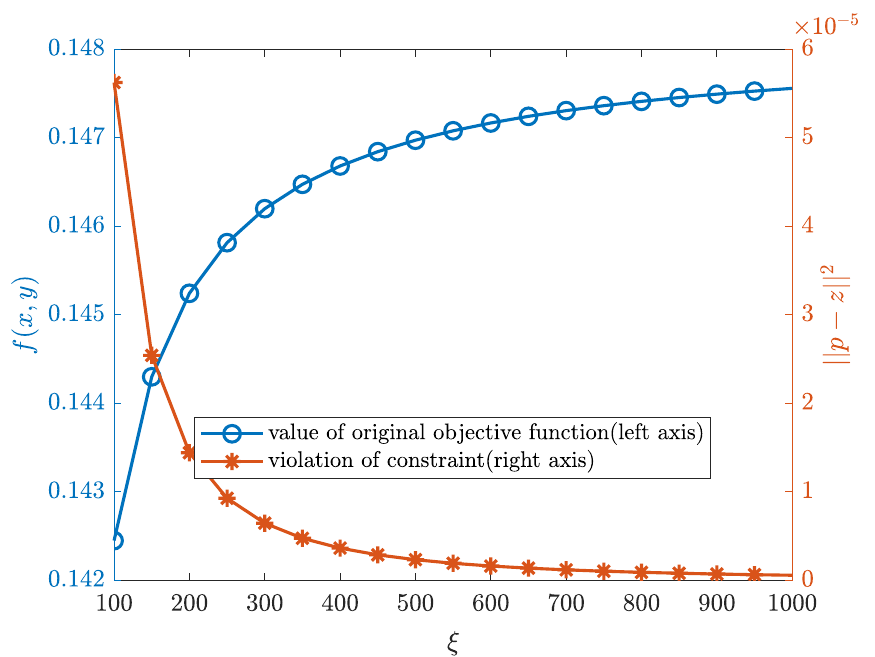}
	\caption{Original objective function $f(x,y)$ (left axis) and penalty term $\left\|p-z\right\|^2$ (right axis) under varying $\xi$ for state $W_3W_3^\top$.}%设置标题，且会自动编号,图在下
	\label{re2}%设置标签
\end{figure} 

The experiments are also excuated on the GHZ states with white noise (See Example \ref{ex2}), where the state is theoretically guaranteed to be bi-PPT when $l>=1$. To investigate the decomposition threshold, we systematically test $x=\left( 0.1,0.2,\cdots,1\right) $ under fixed penalty parameter $\xi=1000$ ensuring strict constraint enforcement.
\begin{example}\label{ex2}
	Consider a three-qubit system with state construction: $\Omega=GHZ_3GHZ_3^\top+lI_8$, $GHZ_3=e_0\otimes e_0\otimes e_0+e_1\otimes e_1\otimes e_1$ where $e_0=\begin{pmatrix}1\\0\end{pmatrix}, e_1=\begin{pmatrix}0\\1\end{pmatrix}$ and $l\in\mathbb{R}_{+}$. The normalized state becomes $\rho=\Omega/\mathrm{tr}(\Omega)$.
\end{example}
Figure \ref{re3} illustrates the algorithm's iteration process across tested $l$ values. For each $l$, we perform 30 trials and select the minimum-error trajectory. Despite comparable iteration counts across all cases, final objective values exhibit monotonic decay with increasing $l$, aligning with theoretical expectations.

Table \ref{bound1} quantifies the decomposition precision, revealing critical threshold behavior: At $l=0.9$, the residual error remains at $6.73e^{-5}$ – three orders above machine precision – confirming the state $GHZ_3GHZ_3^\top+0.9I_8$ cannot satisfy (\ref{proform}). This contrasts sharply with $l=1$ achieving near-machine precision ($2.91e^{-9}$), demonstrating exact decomposition feasibility when $l\ge1$.
\begin{figure}[H]
	\centering
	\includegraphics[width=0.5\textwidth,height=0.26\textheight]{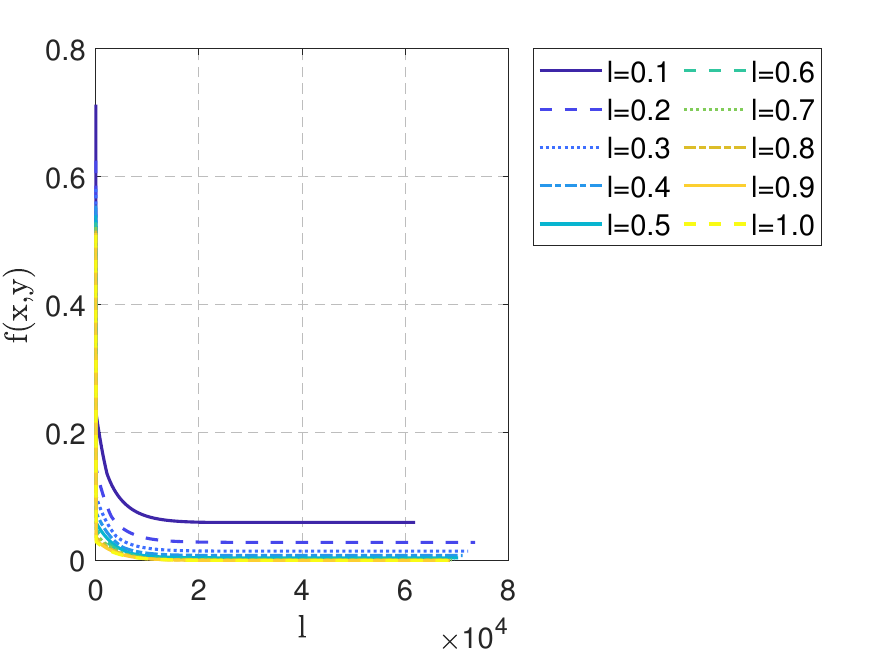}
	\caption{Convergence of GHZ states with different noise coefficients $l$ under $\xi=1000$.}%设置标题，且会自动编号,图在下
	\label{re3}%设置标签
\end{figure}

\begin{table}[H]\setlength{\tabcolsep}{1.5mm}{
		\footnotesize
		\centering
		\caption{Value of original objective function of states $GHZ_3GHZ_3^\top+lI$.}%设置标题，且会自动编号，表在上
		\label{bound1}
		\begin{tabular}{ccccccccccc}%每一列对齐方式：l左对齐，r右对齐，c居中对齐,|为产生表格竖线，||产生双竖线
			%此处即为5列，分别设定
			%编写内容
			\toprule%产生表格横线
			Noise parameter ($l$)&0.1&0.2&0.3&0.4&0.5&0.6&0.7&0.8&0.9&1\\\midrule
			$f(x,y)$&0.0594&0.0281&0.0144&0.0076&0.0040&0.0020&0.0009&0.0003&6.7304$e^{-5}$&2.9110$e^{-9}$\\\bottomrule
	\end{tabular}}
	%具体的设定可以打开相应的宏包说明文件进行查看
	%booktab三线表,longtab跨页长表格，tabu综合表格宏包content...
\end{table}

To further validate the generalizability of our method, we conduct experiments on a ​five-partite quantum system​​ with uniform subsystem dimensions $(3,3,3,3,3)$, representing a challenging high-dimensional scenario with $243\times243$ density matrices. This configuration tests our algorithm's capacity to handle complex multipartite systems beyond symmetric qubit arrangements.
\begin{example}\label{ex3}
	Consider a five-partite quantum system with dimensions $$(d_1,d_2,d_3,d_4,d_5)=(3,3,3,3,3).$$ We construct a noise-perturbed GHZ state: $\Omega=GHZ_5GHZ_5^\top+lI_{243}$, with $l\in\mathbb{R}_{+}$, $GHZ_5=\otimes^5_{k=1}e_k^{\left( 1\right) }+\otimes^5_{k=1}e_k^{\left( 3\right) }$. Here $e_k^{\left( m\right) }$ denotes the canonical basis vector with 1 at the $m$-th position. The normalized state becomes $\rho=\Omega/\mathrm{tr}(\Omega)$.
\end{example}
We first set $l=2$ and obtain the final error. The algorithm configuration employs $\xi=1000$, balancing computational efficiency with constraint enforcement. Figure \ref{reex3} details the convergence characteristics over 130000 iterations. 
\begin{table}[H]
	\footnotesize
	\centering
	\caption{Values of original objective function and violation of states $GHZ_5GHZ_5^\top+lI$.}%设置标题，且会自动编号，表在上
	\label{bound2}
	\begin{tabular}{cccc}%每一列对齐方式：l左对齐，r右对齐，c居中对齐,|为产生表格竖线，||产生双竖线
		%此处即为5列，分别设定
		%编写内容
		\toprule%产生表格横线
		Noise parameter ($l$)&0.9&1.1&2\\\midrule
		$f(x,y)$&3.0935$e^{-7}$&8.0043$e^{-10}$&8.0321$e^{-10}$\\\midrule
		$\left\|p-z\right\|^2$&6.3122$e^{-13}$&1.9962$e^{-16}$&2.2957$e^{-16}$\\\bottomrule
	\end{tabular}
	%具体的设定可以打开相应的宏包说明文件进行查看
	%booktab三线表,longtab跨页长表格，tabu综合表格宏包content...
\end{table}

\begin{figure}[H]
	\centering
	\begin{subfigure}[htpb]{0.325\textwidth} % 修复1：添加必要的宽度参数
		\centering % 修复2：添加居中指令
		\includegraphics[width=\linewidth]{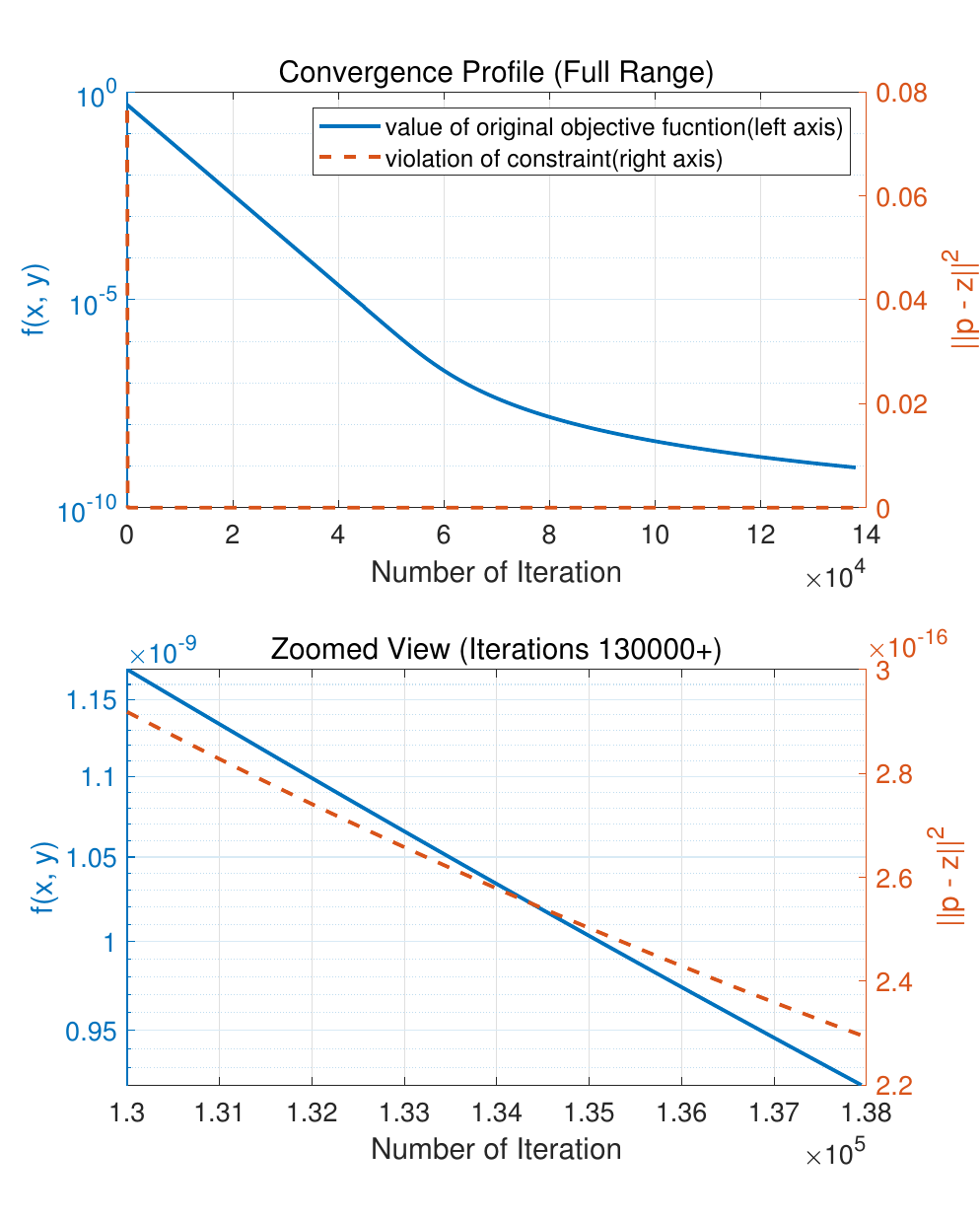} % 修复3：使用\linewidth
		\caption{Case: $l=2$.}
		\label{re4}
	\end{subfigure}
	\hspace{0.1em}
	\begin{subfigure}[htpb]{0.325\textwidth} % 修复1：添加必要的宽度参数
		\centering % 修复2：添加居中指令
		\includegraphics[width=\linewidth]{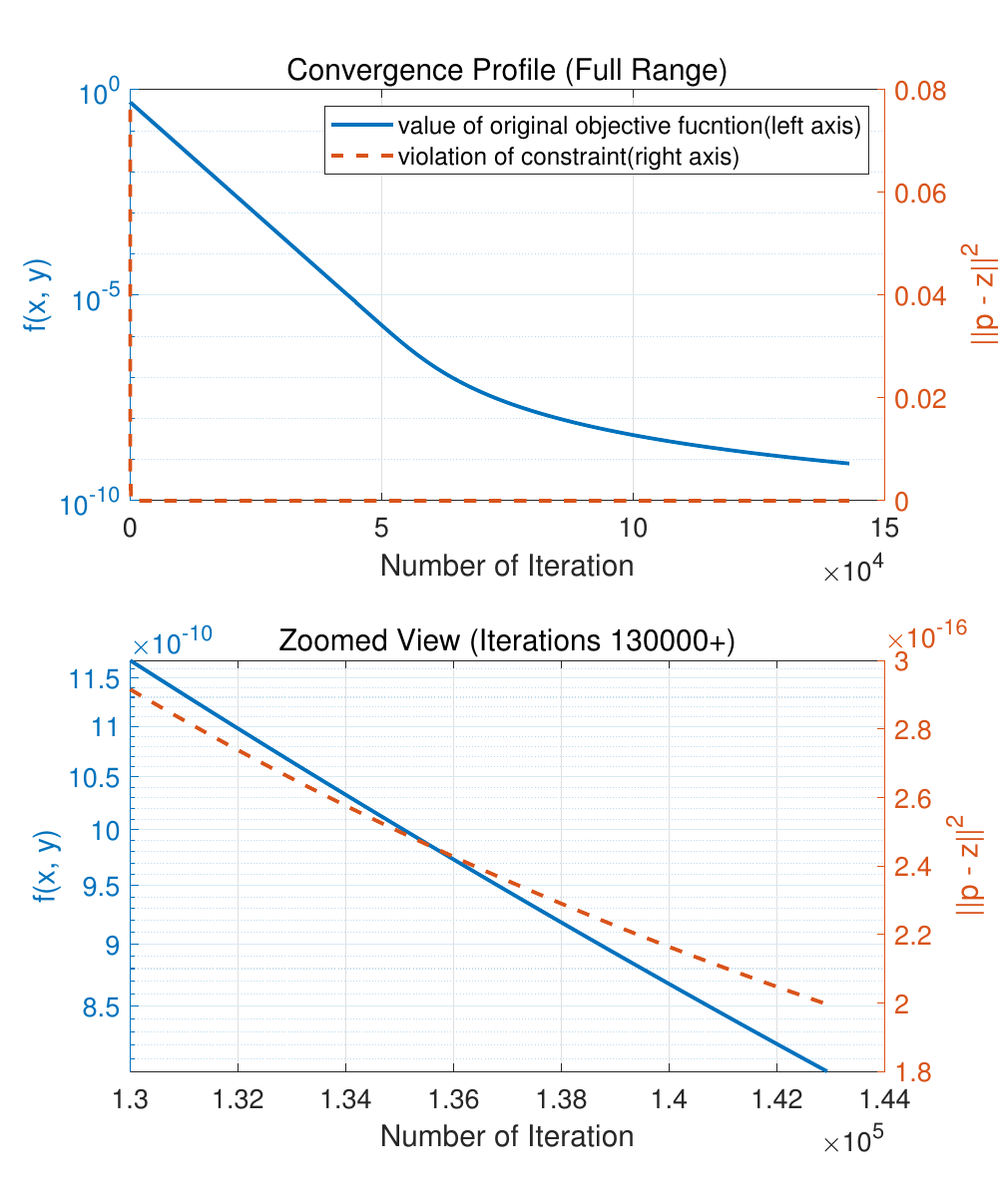} % 修复3：使用\linewidth
		\caption{Case: $l=1.1$.}
		\label{re6}
	\end{subfigure}
	\hfill
	\begin{subfigure}[htpb]{0.325\textwidth} % 修复1：添加必要的宽度参数
		\centering % 修复2：添加居中指令
		\includegraphics[width=\linewidth]{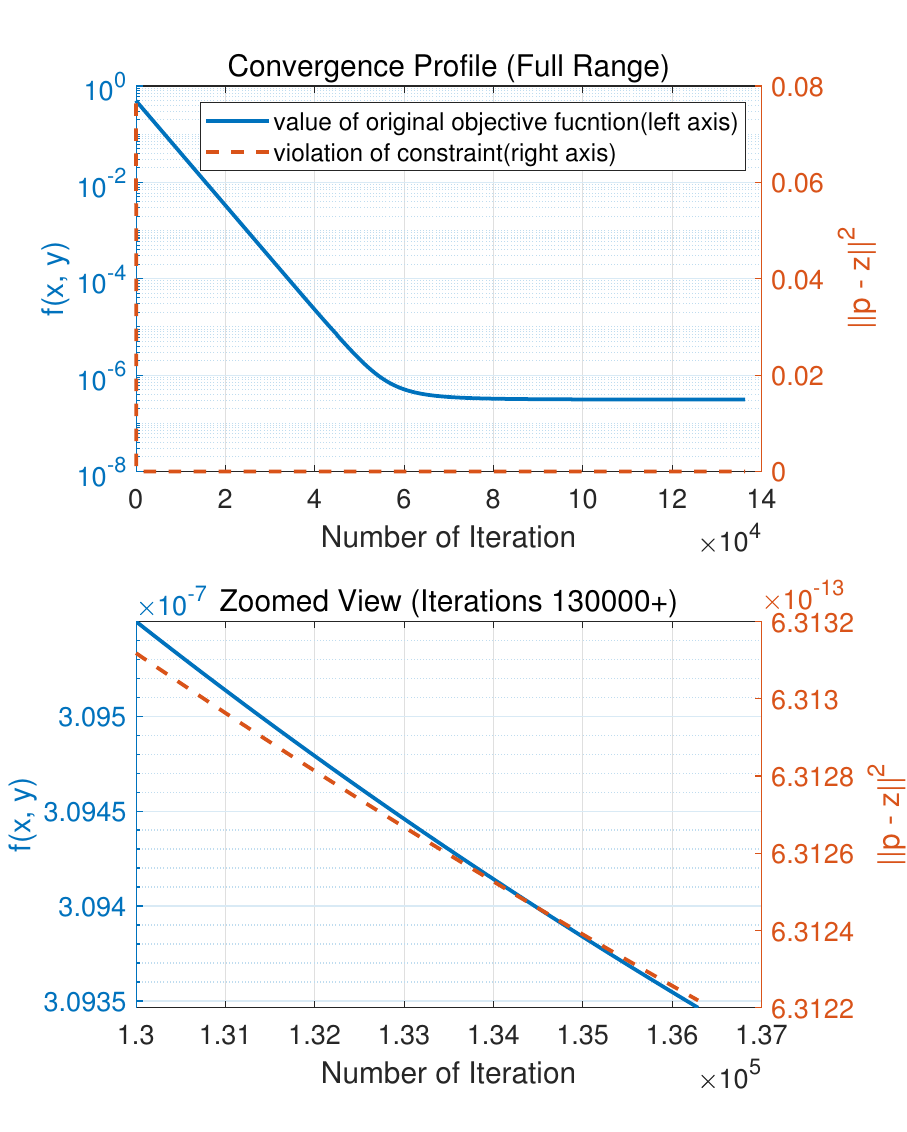} % 修复3：使用\linewidth
		\caption{Case: $l=0.9$.}%设置标题，且会自动编号,图在下
		\label{re5}%设置标签
	\end{subfigure}
	\caption{Convergence of GHZ states in five-partite system with different parameters $l$.} % 修复5：添加有意义的标题
	\label{reex3}
\end{figure}

The results confirm robust performance of our method in high-dimensional settings. Notably, the violation magnitude approaches machine epsilon, indicating near-exact constraint satisfaction. It can be found that when $l=0.9$, the GHZ state can not be decmoposed as the form (\ref{proform}). The original objective function value can reach $e^{-10}$ under $l=1.1$ and $l=2$. Now We give the last example, which is an unweighted form of multiGHZ state. 
\begin{example}\label{ex4}
	We construct a noise-perturbed multiGHZ state: $\Omega=mGHZ_5mGHZ_5^\top+lI_{243}$, with $l\in\mathbb{R}_{+}$, $mGHZ_5=m\otimes^5_{k=1}e_k^{\left( 1\right) }+n\otimes^5_{k=1}e_k^{\left( 2\right) }+s\otimes^5_{k=1}e_k^{\left( 3\right) }$. The normalized state becomes $\rho=\Omega/\mathrm{tr}(\Omega)$.
\end{example}
We search this example for identifying the influence of the both noise and state coefficients on the decomposition. Two cases are considered: for fixed $m=n=1,s=5$, we change $l$, and for fixed $m=n=1,l=5$, we change $s$. The final results are shown in Table \ref{lastex}, which indicates that when $l=s=5$, the approximation error can reach $e^{-10}$, and if one is fixed as $5$, the other one change $1$, then the state can not be decomposed as the convex combination of bi-PPT states. However, the comparision of 6.4236$e^{-7}$ and 1.0101$e^{-6}$ demonstrates that the decrease of $l$ has more influence for the decomposition.
\begin{table}[H]
	\footnotesize
	\centering
	\caption{Values of original objective function and violation of states $mGHZ_5mGHZ_5^\top+lI$.}%设置标题，且会自动编号，表在上
	\label{lastex}
	\begin{tabular}{ccccccc}%每一列对齐方式：l左对齐，r右对齐，c居中对齐,|为产生表格竖线，||产生双竖线
		%此处即为5列，分别设定
		%编写内容
		\toprule%产生表格横线
		$l$&5&5&5&6&5&4\\\midrule
		$s$&3&4&6&5&5&5\\\midrule
		$f(x,y)$&8.0052$e^{-10}$&8.0157$e^{-10}$&6.4236$e^{-7}$&8.4163$e^{-10}$&8.0042$e^{-10}$&1.0101$e^{-6}$\\\midrule
		$\left\|p-z\right\|^2$&$2.1102e^{-16}$&1.9981$e^{-16}$&1.3008$e^{-12}$&$2.0989e^{-16}$&1.9962$e^{-16}$&2.0463$e^{-12}$\\\bottomrule
	\end{tabular}
	%具体的设定可以打开相应的宏包说明文件进行查看
	%booktab三线表,longtab跨页长表格，tabu综合表格宏包content...
\end{table}

\section{Conclusion}\label{con}
In conclusion, this work presents a successful reformulation of the convex combination of bi-PPT states within a novel optimization framework, establishing the practical viability of the LPADMM for addressing this class of challenges. Through rigorous analysis based on a penalized approximation model, we prove that the sequence generated by LPADMM converges globally to a stationary point, with an iteration complexity of $O(1/\epsilon^2)$ guaranteed for identifying an $\epsilon$-stationary solution. Comprehensive numerical experiments across diverse decomposition paradigms have further validated the effectiveness of our approach, demonstrating its capability to produce satisfactory approximations of the target quantum state.

Notwithstanding these advancements, it is important to acknowledge that the attainment of a stationary point does not inherently guarantee convergence to the global minimum. The methodological constraints arising from the indicator function's non-convex nature currently prevent direct application of our framework to the original problem formulation. Consequently, the development of more precise approximation techniques for handling such non-convexities remains an open challenge for future research.

\section{Acknowledgments}
This work was supported by the NSFC (Grants No.12471427 and 12131004), and the R\&D proiect of Pazhou Lab (Huangpu)(2023K0604).

\appendix
\section{Marginal Smoothness and Convexity}\label{app1}
\begin{lemma}\label{inner}
	For any symmetric positive definite matrices $A$, $B$ $\in\mathbb{S}^n_{++}$, their Frobenius inner product satisfies $\left\langle A,B\right\rangle\ge0 $.
\end{lemma}
\begin{proof}
	Let $A=PP^\top$, and $B=QQ^\top$ be Cholesky decompositions. Then:
	\begin{align*}
		\left\langle A,B\right\rangle=tr\left( A^\top B\right) =tr\left( P^\top PQQ^\top\right) =tr\left( Q^\top P^\top PQ\right) =tr\left( PQ\left( PQ\right) ^\top\right) =\left\|PQ\right\|^2\ge0.
	\end{align*}
\end{proof}

\begin{lemma}\label{convexsmooth}
	Under constraints $a+b+c=1, a\ge0, b\ge0, c\ge0$ and $\mathrm{mat}\left( \alpha\right) \succeq0, \mathrm{mat}\left( \beta\right) \succeq0, \mathrm{mat}\left( \gamma\right) \succeq0$, the function $f\left( x,y\right) $ in problem ( \ref{problem}) is:
	\begin{itemize}
		\item[\rm(i)] 1-smooth with respect to $x$,
		\item[\rm(ii)] marginally convex.
	\end{itemize}
\end{lemma}
\begin{proof}
	The Hessian of $f\left( x,y\right) $. has block structure:
	\begin{align*}
		\nabla^2_{\left( x,y\right) } f\left( x,y\right)  = \begin{pmatrix} M & 0 \\ 0 & N \end{pmatrix},
	\end{align*}
	where
	$$M=\begin{pmatrix}
		a^2I&abI&acI\\
		baI&b^2I&bcI\\
		caI&cbI&c^2I\\
	\end{pmatrix}, 
	N=\begin{pmatrix}
		\left\|\alpha\right\|^2&\left\langle \alpha,\beta\right\rangle& \left\langle \alpha,\gamma\right\rangle\\
		\left\langle\alpha,\beta\right\rangle&\left\|\beta\right\|^2& \left\langle \beta,\gamma\right\rangle\\
		\left\langle\gamma,\alpha\right\rangle&\left\langle \gamma,\beta\right\rangle& \left\|\gamma\right\|\\
	\end{pmatrix}.$$
	
	\textit{Smoothness:}  The minimum eigenvalue of $M$ is 0, and the maximum eigenvalue of $M$ is $\left( a^2+b^2+c^2\right) \leqslant \left( a+b+c\right) ^2=1$ for which $a\ge0,b\ge0,c\ge0$. It indicates that $f$ is $L$-smooth with respect to variable $x$ and $L\leqslant1$ for any fixed $y$. \\
	
	\textit{Convexity:} For matrix $N$, the innerproduct terms satisfy $\left\langle \alpha,\beta\right\rangle\ge0$, $\left\langle \alpha,\gamma\right\rangle\ge0$, $\left\langle \gamma,\beta\right\rangle\ge0$ with the support of Lemma \ref{inner} and $\left\langle \alpha,\beta\right\rangle^2\leqslant\left\|\alpha\right\|_F^2\left\|\beta\right\|_F^2 $ is true according to Cauchy-Schwarz inequality. 
	
	Now the eigenvalues of matrix $N$ are calculated below. Here we do not consider the degenerate situation which means $d:=\left\|\alpha\right\|^2>0, e:=\left\|\beta\right\|^2>0, f:=\left\|\gamma\right\|^2>0$.
	$$N=\begin{pmatrix}
		d^2&de-\epsilon_1&df-\epsilon_2\\
		ed-\epsilon_1&e^2&ef-\epsilon_3\\
		fd-\epsilon_2&fe-\epsilon_3&f^2\\
	\end{pmatrix},\quad \forall \epsilon_1\in[0,de], \epsilon_2\in[0,df], \epsilon_3\in[0,ef].$$
	The characteristic equation is as follows:
	\begin{align*}
		&-\theta^3+\left( d^2+e^2+f^2\right) \theta^2+\left[ \left( de-\epsilon_1\right) ^2+\left( df-\epsilon_2\right) ^2+\left( ef-\epsilon_3\right) ^2-\left( de\right) ^2-\left( df\right) ^2-\left( ef\right) ^2\right]\theta\\
		&+\left[ d^2e^2f^2+2\left( de-\epsilon_1\right) \left( df-\epsilon_2\right) \left( fe-\epsilon_3\right)-f^2\left( de-\epsilon_1\right) ^2
		-e^2\left( df-\epsilon_2\right) ^2-d^2\left( fe-\epsilon_3\right) ^2\right].
	\end{align*}
	The relationships of the roots can be obtained by Cardano's formula, which are
	\begin{align*}
		\theta_1+\theta_2+\theta_3=-\frac{p_2}{p_1},\\
		\frac{1}{\theta_1}+\frac{1}{\theta_2}+\frac{1}{\theta_3}=-\frac{p_3}{p_4},\\
		\theta_1\theta_2\theta_3=-\frac{p_4}{p_1},
	\end{align*}
	where
	\begin{align*}
		&p_1=-1,\\
		&p_2=d^2+e^2+f^2,\\
		&p_3=\left( de-\epsilon_1\right) ^2+\left( df-\epsilon_2\right) ^2+\left( ef-\epsilon_3\right) ^2-\left( de\right) ^2-\left( df\right) ^2-\left( ef\right) ^2,\\
		&p_4=d^2e^2f^2+2\left( de-\epsilon_1\right) \left( df-\epsilon_2\right) \left( fe-\epsilon_3\right) -f^2\left( de-\epsilon_1\right) ^2-e^2\left( df-\epsilon_2\right) ^2-d^2\left( fe-\epsilon_3\right) ^2.
	\end{align*}
	It can be obtained that $p_3\leqslant0$.
	
	Then, we prove that $p_4\ge0$. $p_4$ can be seen as the function respect to variables $\epsilon_1, \epsilon_2, \epsilon_3$. The Hessian matrix of $p_4$ is negative, which means the function $p_4$ is strongly concave. The gradient of $p_4$ is
	\begin{align*}
		\frac{\partial p_4}{\partial \epsilon_1}=-2\left( df-\epsilon_2\right) \left( ef-\epsilon_3\right) +2f^2\left( de-\epsilon_1\right) .
	\end{align*}
	When $\epsilon_1=de, \epsilon_2=df, \epsilon_3=ef$, $\frac{\partial p_4}{\partial\epsilon}=0$, it is the minimum point of $p_4$. With the Mangasarian-Fromovitz (M-F)  conditions of objective function, the minimum point must be KKT point. The KKT point can be calculated as $\left( 0,0,0\right) $ and $\left( de,df,ef\right) $ within the box constraints. Therefore, the minimal value is 0, which means $p_4\ge0$. Finally, $\theta_1, \theta_2, \theta_3\ge0$ can be deducted by $\theta_1+\theta_2+\theta_3=p_2\ge0$; $\theta_1\theta_2\theta_3=p_4\ge0$; $\frac{1}{\theta_1}+\frac{1}{\theta_2}+\frac{1}{\theta_3}\ge0$.
\end{proof}

\section{Convergence Analysis}\label{convergence}
The optimization problem under consideration is nonconvex and nonseparable with linear equality and set constraints, making it challenging to guarantee global optimality from arbitrary initializations. We establish that our LPADMM algorithm converges to a stationary point with iteration complexity by analyzing first-order optimality conditions.

\begin{theorem}[Optimality conditions]\label{thm:opt_cond} A tuple $\left( y^*,x^*,p^*,z^*,\lambda^*\right) $ constitutes a stationary point of problem (\ref{proapp}) if there exists a Lagrange multiplier $\lambda^*$ satisfying:
	\begin{align}
		\begin{cases}
			0\in\nabla_{y}f\left( x^*,y^*\right) +\partial\delta_{\mathcal{Y}}\left( y^*\right) ,\\
			0\in\nabla_{x}f\left( x^*,y^*\right) +\partial\delta_{\mathcal{X}}\left( x^*\right) +A^\top\lambda^*,\\
			0\in\partial\delta\left( p^*\right) +\xi\left( p^*-z^*\right) ,\\
			\xi\left( z^*-p^*\right) -\lambda^*=0,\\
			Ax^*-z^*=0,
		\end{cases}
	\end{align}
then $\left( y^*,x^*,p^*,z^*,\lambda^*\right) $ is a stationary point of problem (\ref{proapp}).
\end{theorem}
\begin{proof}
The result follows directly from subdifferential calculus rules \cite{nesterov2018lectures} and the definition of stationary points in constrained optimization.
\end{proof}
We establish the descent properties of the augmented Lagrangian $\mathcal{L}_{\eta}$ through successive variable updates.
\begin{lemma}[$y$-update descent]\label{lem:y_descent}
	Let $\left\lbrace\left( y^{k},x^k,p^k,z^k,\lambda^k\right)  \right\rbrace$ be the sequence generated by LPADMM, then 
	\begin{align}\label{yde}
		\mathcal{L}_{\eta}\left( y^{k+1},x^k,p^k,z^k,\lambda^k\right) -\mathcal{L}_{\eta}\left( y^{k},x^k,p^k,z^k,\lambda^k\right) \leqslant-\mu_1\left\|y^{k+1}-y^k\right\|^2.
	\end{align}
\end{lemma}
\begin{proof}
	Combining (\ref{eq:case1}) with the convexity of $f\left( x^k,y\right) $ and $\delta_{\mathcal{Y}}\left( y\right) $, we have
	\begin{align*}
		\mathcal{L}_{\eta}\left( y^{k+1},x^k,p^k,z^k,\lambda^k\right) -\mathcal{L}_{\eta}\left( y^{k},x^k,p^k,z^k,\lambda^k\right) \leqslant\partial^\top\delta_{\mathcal{Y}}\left( y^{k+1}\right) \left( y^{k+1}-y^{k}\right) \leqslant-\mu_1\left\|y^{k+1}-y^k\right\|^2.
	\end{align*}
\end{proof}
\begin{lemma}[Joint $x-p-z$ update descent]\label{lem:xpz_descent}
	Let $\left\lbrace \left( y^{k},x^k,p^k,z^k,\lambda^k\right) \right\rbrace$ be the sequence generated by LPADMM, then 
	\begin{align}\label{xpzde}
		&\mathcal{L}_{\eta}\left( y^{k+1},x^{k+1},p^{k+1},z^{k+1},\lambda^k\right) -\mathcal{L}_{\eta}\left( y^{k+1},x^k,p^k,z^k,\lambda^k\right) \notag\\
		\leqslant&-\frac{\eta}{2}\left\|z^{k+1}-z^k\right\|^2-\frac{\eta}{2}\left\|x^{k+1}-x^k\right\|^2_{A^\top A}-\left( \mu_2-L\right) \left\|x^{k+1}-x^k\right\|^2-\mu_3\left\|p^{k+1}-p^k\right\|^2.
	\end{align}
\end{lemma}
\begin{proof}
	With (\ref{eq:case2}), the following inequalities can be obtained.
	\begin{align}
		&\mathcal{L}_{\eta}\left( y^{k+1},x^{k+1},p^k,z^k,\lambda^k\right) -\mathcal{L}_{\eta}\left( y^{k+1},x^k,p^k,z^k,\lambda^k\right) \notag\\
		=\:&\delta_{\mathcal{X}}\left( x^{k+1}\right) -\delta_{\mathcal{X}}\left( x^k\right) +f\left( x^{k+1},y^{k+1}\right) -f\left( x^k,y^{k+1}\right) +\left\langle \lambda^k,Ax^{k+1}-Ax^k\right\rangle\notag\\ &+\frac{\eta}{2}\left\|Ax^{k+1}-z^{k}\right\|^2-\frac{\eta}{2}\left\|Ax^{k}-z^{k}\right\|^2\notag\\
		\leqslant\:&\partial^\top\delta_{\mathcal{X}}\left( x^{k+1}\right) \left( x^{k+1}-x^{k}\right) +f\left( x^{k+1},y^{k+1}\right) -f\left( x^k,y^{k+1}\right) +\left\langle \lambda^k,Ax^{k+1}-Ax^k\right\rangle\notag\\
		&+\frac{\eta}{2}\left\|Ax^{k+1}-z^{k}\right\|^2-\frac{\eta}{2}\left\|Ax^{k}-z^{k}\right\|^2\notag\\
		=&\left\langle \nabla_x^\top f\left( x^k,y^{k+1}\right) -A^\top\lambda^k-\eta A^\top\left( Ax^{k+1}-z^k\right) -\mu_2\left( x^{k+1}-x^k\right) ,x^{k+1}-x^k\right\rangle\notag\\
		&+\left\langle\lambda^k,Ax^{k+1}-Ax^k\right\rangle+\frac{\eta}{2}\left\|Ax^{k+1}-z^{k}\right\|^2-\frac{\eta}{2}\left\|Ax^{k}-z^{k}\right\|^2+f\left( x^{k+1},y^{k+1}\right) -f\left( x^k,y^{k+1}\right) \notag\\
		=&\left\langle -A^\top\lambda^k-\eta^k A^\top\left( Ax^{k+1}-z^{k+1}-z^k+z^{k+1}\right) -\mu_2\left( x^{k+1}-x^k\right) ,x^{k+1}-x^k\right\rangle\notag\\
		&+\left\langle\lambda^k,Ax^{k+1}-Ax^k\right\rangle+\frac{\eta}{2}\left\|Ax^{k+1}-z^{k}\right\|^2-\frac{\eta}{2}\left\|Ax^{k}-z^{k}\right\|^2+f\left( x^{k+1},y^{k+1}\right) -f\left( x^k,y^{k+1}\right) \notag\\
		&-\left\langle \nabla_x^\top f\left( x^{k+1},y^{k+1}\right) ,x^{k+1}-x^k\right\rangle+\left\langle \nabla_x^\top f\left( x^{k+1},y^{k+1}\right) -\nabla_x^\top f\left( x^{k},y^{k+1}\right) ,x^{k+1}-x^k\right\rangle\notag\\
		\leqslant&-\left\langle\lambda^{k+1}-\lambda^k,Ax^{k+1}-Ax^k\right\rangle+\frac{\eta}{2}\left\|Ax^{k+1}-z^{k}\right\|^2-\frac{\eta}{2}\left\|Ax^{k}-z^{k}\right\|^2\notag\\
		&-\eta\left\langle z^{k+1}-z^k,A^{x+1}-Ax^k\right\rangle-\left( \mu_2-L\right) \left\|x^{k+1}-x^k\right\|^2.\label{ineq1}
	\end{align}
	The first inequality is obtained by the convexity of $\delta_{\mathcal{X}}\left( x\right) $, and the last inequality is obtained by the 1-smoothness of $f\left( x,y\right) $ with respect to variable $x$ for any fixed $y$.
	
	Here, we set $g\left( p,z\right) =\frac{\xi}{2}\left\|p-z\right\|^2$. By (\ref{eq:case3}), we have
	\begin{align}
		&\mathcal{L}_{\eta}\left( y^{k+1},x^{k+1},p^{k+1},z^k,\lambda^k\right) -\mathcal{L}_{\eta}\left( y^{k+1},x^{k+1},p^k,z^k,\lambda^k\right) \notag\\
		=\:&\delta_{\mathcal{Z}}\left( p^{k+1}\right) -\delta_{\mathcal{Z}}\left( p^k\right) +g\left( p^{k+1},z^k\right) -g\left( p^k,z^k\right) \notag\\
		\leqslant\:&[\partial\delta_{\mathcal{Z}}\left( p^{k+1}\right) +\nabla_{p}g\left( p^{k+1},z^k\right) ]^\top\left( p^{k+1}-p^k\right) -\mu_3\left\|p^{k+1}-p^k\right\|^2.\label{ineq2}
	\end{align}
	The update of $z$ (\ref{eq:case4}) yields
	\begin{align}
		&\mathcal{L}_{\eta}\left( y^{k+1},x^{k+1},p^{k+1},z^{k+1},\lambda^k\right) -\mathcal{L}_{\eta}\left( y^{k+1},x^{k+1},p^{k+1},z^k,\lambda^k\right) \notag\\
		=\:&g\left( p^{k+1},z^{k+1}\right) -\left\langle \lambda^k,z^{k+1}\right\rangle+\frac{\eta}{2}\left\|Ax^{k+1}-z^{k+1}\right\|^2-g\left( p^{k+1},z^k\right) +\left\langle \lambda^k,z^k\right\rangle-\frac{\eta}{2}\left\|Ax^{k+1}-z^{k}\right\|^2\notag\\
		\leqslant\:&\nabla_{z}g\left( p^{k+1},z^{k+1}\right) ^\top\left( z^{k+1}-z^k\right) -\left\langle \lambda^k,z^{k+1}-z^k\right\rangle+\frac{\eta}{2}\left\|Ax^{k+1}-z^{k+1}\right\|^2-\frac{\eta}{2}\left\|Ax^{k+1}-z^{k}\right\|^2\notag\\
		=&\left\langle \lambda^k-\eta \left( -Ax^{k+1}+z^{k+1}\right) ,z^{k+1}-z^k\right\rangle-\left\langle \lambda^k,z^{k+1}-z^k\right\rangle\notag\\
		&+\frac{\eta}{2}\left\|Ax^{k+1}-z^{k+1}\right\|^2-\frac{\eta}{2}\left\|Ax^{k+1}-z^{k}\right\|^2\notag\\
		=&\left\langle \lambda^{k+1}-\lambda^k,z^{k+1}-z^k\right\rangle+\frac{\eta}{2}\left\|Ax^{k+1}-z^{k+1}\right\|^2-\frac{\eta}{2}\left\|Ax^{k+1}-z^{k}\right\|^2.\label{ineq3}
	\end{align}
	Since $g\left( p,z\right) $ is convex, (\ref{lem:xpz_descent}) can be derived by adding (\ref{ineq1}), (\ref{ineq2}) and (\ref{ineq3}).
\end{proof}

\begin{lemma}[Dual variable bound]\label{lem:dual_bound}
	Let $\left\lbrace \left( y^{k},x^k,p^k,z^k,\lambda^k\right) \right\rbrace$ be the sequence generated by LPADMM, then 
	\begin{align}\label{lamde}
		&\mathcal{L}_{\eta}\left( y^{k+1},x^{k+1},p^{k+1},z^{k+1},\lambda^{k+1}\right) -\mathcal{L}_{\eta}\left( y^{k+1},x^{k+1},p^{k+1},z^{k+1},\lambda^k\right) \notag\\
		\leqslant&\frac{2\xi^2}{\eta}\left( \left\|p^{k+1}-p^k\right\|^2+\left\|z^{k+1}-z^k\right\|^2\right) .
	\end{align}
\end{lemma}
\begin{proof}
	From (\ref{eq:case5}), we have
	\begin{align}\label{lamequ}
		&\mathcal{L}_{\eta}\left( y^{k+1},x^{k+1},p^{k+1},z^{k+1},\lambda^{k+1}\right) -\mathcal{L}_{\eta}\left( y^{k+1},x^{k+1},p^{k+1},z^{k+1},\lambda^k\right) =\frac{1}{\eta}\left\| \lambda^{k+1}-\lambda^k\right\|^2.
	\end{align}
	Then, we estimate the positive term $\frac{1}{\eta}\left\| \lambda^{k+1}-\lambda^k\right\|^2$. From (\ref{eq:case4}), the following equality can be obtained.
	\begin{align*}
		\lambda^{k+1}=\xi\left( p^{k+1}-z^{k+1}\right) .
	\end{align*}
	By triangle inequality for norms, we get
	\begin{align}\label{lambda}
		\left\|\lambda^{k+1}-\lambda^k\right\|^2\leqslant2\xi^2\left( \left\|p^{k+1}-p^k\right\|^2+\left\|z^{k+1}-z^k\right\|^2\right) .
	\end{align}
	The result can be obtained by substituting (\ref{lambda}) into (\ref{lamequ}).
\end{proof}

\begin{proposition}[Sufficient descent of $\mathcal{L}_{\eta}$]\label{prop:sufficient_descent}
	If $\eta>2\xi$, $\mu_2>L$, $\mu_3>\frac{2\xi^2}{\eta}$, and the sequence $\left\lbrace \left( y^{k},x^k,p^k,z^k,\lambda^k\right) \right\rbrace$ is generated by LPADMM, then $\mathcal{L}_{\eta}$ satisfies:	
	\begin{align}\label{final}
		&\mathcal{L}_{\eta}\left( y^{k+1},x^{k+1},p^{k+1},z^{k+1},\lambda^{k+1}\right) -\mathcal{L}_{\eta}\left( y^{k},x^k,p^k,z^k,\lambda^k\right) \notag\\
		\leqslant&-\left( \frac{\eta}{2}-\frac{2\xi^2}{\eta}\right) \left\|z^{k+1}-z^k\right\|^2-\frac{\eta}{2}\left\|x^{k+1}-x^k\right\|^2_{A^\top A}-\mu_1\left\|y^{k+1}-y^k\right\|^2\notag\\
		&-\left( \mu_2-L\right) \left\|x^{k+1}-x^k\right\|^2-\left( \mu_3-\frac{2\xi^2}{\eta}\right) \left\|p^{k+1}-p^k\right\|^2.
	\end{align}
\end{proposition}
\begin{proof}
	Summing the estimates from (\ref{yde}), (\ref{xpzde}) and (\ref{lamde}) under the specified parameter conditions yields the composite descent inequality.
\end{proof}

\begin{proposition}[Boundness of variables]\label{prop:boundedness}
	Under the conditions of Proposition \ref{prop:sufficient_descent}:
	\begin{itemize}
	\item[(i)]$\mathcal{L}_{\eta}\left( y^{k},x^k,p^k,z^k,\lambda^k\right) $ is lower bounded for all $k\in\mathbb{N}$ and converges as $k\to\infty$.
	\item[(ii)]The sequence $\left\lbrace \left( y^{k},x^k,p^k,z^k,\lambda^k\right) \right\rbrace$ is bounded.
	\end{itemize}
\end{proposition}
\begin{proof}
	(i) It is obvious that $f\left( x^k,y^k\right) +\delta_{\mathcal{Z}}\left( p^k\right) +\delta_{\mathcal{Y}}\left( y^k\right) +\delta_{\mathcal{X}}\left( x^k\right) \ge-\infty$. Then, we have
	\begin{align*}
		\mathcal{L}_{\eta}\left(  y^{k},x^k,p^k,z^k,\lambda^k\right) =&f\left( x^k,y^k\right) +\delta_{\mathcal{Z}}\left( p^k\right) +\delta_{\mathcal{Y}}\left( y^k\right) +\delta_{\mathcal{X}}\left( x^k\right) +\frac{\xi}{2}\left\|p^k-z^k\right\|^2\\
		&+\left\langle\lambda^k,Ax^k-z^k \right\rangle+\frac{\eta}{2}\left\|Ax^k-z^k\right\|^2\\
		=&f\left( x^k,y^k\right) +\delta_{\mathcal{Z}}\left( p^k\right) +\delta_{\mathcal{Y}}\left( y^k\right) +\delta_{\mathcal{X}}\left( x^k\right) \\
		&+\frac{\xi}{2}\left\|p^k-z^k\right\|^2+\xi\left\langle p^k-z^k,Ax^k-z^k \right\rangle+\frac{\eta}{2}\left\|Ax^k-z^k\right\|^2\notag\\
		\ge& f\left( x^k,y^k\right) +\delta_{\mathcal{Z}}\left( p^k\right) +\delta_{\mathcal{Y}}\left( y^k\right) +\delta_{\mathcal{X}}\left( x^k\right) +\frac{\xi}{2}\left\|Ax^k+p^k-2z^k\right\|^2\\
		&+\frac{\eta-\xi}{2}\left\|Ax^k-z^k\right\|^2>-\infty\notag.
	\end{align*} 
	(ii) From the monotonicity of $\mathcal{L}_{\eta}$ and (i), $\mathcal{L}_{\eta}\left( y^{k},x^k,p^k,z^k,\lambda^k\right) $ is upper bounded by $\mathcal{L}_{\eta}\left( y^{0},x^0,p^0,z^0,\lambda^0\right) $ and so are $f\left( x^k,y^k\right) +\delta_{\mathcal{Z}}\left( p^k\right) +\delta_{\mathcal{Y}}\left( y^k\right) +\delta_{\mathcal{X}}\left( x^k\right) +\frac{\xi}{2}\left\|p^k-z^k+Ax^k-z^k\right\|^2$ and $\frac{\eta}{2}\left\|Ax^k-z^k\right\|^2$. It is known that $$\left\|p^k-z^k+Ax^k-z^k\right\|^2\ge\left\|p^k-z^k\right\|^2-\left\|Ax^k-z^k\right\|^2.$$ Due to the boundness of $\left\|Ax^k-z^k\right\|^2$, the objective function in (\ref{proapp}) is upper bounded. Because the objective function is coercive, $\left\lbrace x^k\right\rbrace$, $\left\lbrace y^k\right\rbrace$ and $\left\lbrace p^k\right\rbrace$ are bounded. Therefore, $\left\lbrace Ax^k\right\rbrace$ is bounded which indicates $\left\lbrace z^k\right\rbrace$ is bounded. By (\ref{lambda}), we have $\left\lbrace \lambda^k\right\rbrace$ is bounded.
\end{proof}
Next, we will show the boundness of the subgradient of $\mathcal{L}_{\eta}\left( y^{k+1},x^{k+1},p^{k+1},z^{k+1},\lambda^{k+1}\right) $.
\begin{proposition}[Boundness of the subgradient]\label{prop:subgrad_bound}
	Let $\left\lbrace \left( y^{k},x^{k},p^{k},z^{k},\lambda^{k}\right) \right\rbrace$ be the sequence generated by LPADMM. Under the constraints in Proposition\ref{prop:sufficient_descent}, there exists a constant $C^*$ and $D^{k+1}\in\partial\mathcal{L}_{\eta}\left( y^{k+1},x^{k+1},p^{k+1},z^{k+1},\lambda^{k+1}\right) $ such that
	\begin{align*}
		\left\|D^{k+1}\right\|\leqslant C^*\left( \left\|y^{k+1}-y^k\right\|^2+\left\|x^{k+1}-x^k\right\|^2+\left\|p^{k+1}-p^k\right\|^2+\left\|z^{k+1}-z^k\right\|^2\right) .
	\end{align*}
\end{proposition}
\begin{proof}
	From the iteration steps, we get
	\begin{align*}
		D^{k+1}=\begin{pmatrix}
			\nabla_{y}f\left( x^{k+1},y^{k+1}\right) -\nabla_{y}f\left( x^{k},y^{k+1}\right) -\mu_1\left( y^{k+1}-y^k\right) \\
			\nabla_{x}f\left( x^{k+1},y^{k+1}\right) -\nabla_{x}f\left( x^{k},y^{k+1}\right) +A^\top\left( \lambda^{k+1}-\lambda^k\right) +\eta A^\top\left( z^k-z^{k+1}\right) -\mu_2\left( x^{k+1}-x^k\right) \\
			\xi\left( z^k-z^{k+1}\right) -\mu_3\left( p^{k+1}-p^k\right) \\
			\lambda^k-\lambda^{k+1}\\
			\frac{1}{\eta}\left( \lambda^{k+1}-\lambda^{k}\right) 
		\end{pmatrix}.
	\end{align*}
	Then, we calculate the bound of subgradient above. From (\ref{lem:dual_bound}), we know that
	\begin{align*}
		\left\|\lambda^k-\lambda^{k+1}\right\|\leqslant\xi\left( \left\|p^k-p^{k+1}\right\|+\left\|z^k-z^{k+1}\right\|\right) .
	\end{align*}
	By the 1-smoothness of $f$ with respect to $x$ for every fixed $y$, we have
	\begin{align*}
		\left\|\nabla_{x}f\left( x^{k+1},y^{k+1}\right) -\nabla_{x}f\left( x^{k},y^{k+1}\right) \right\|\leqslant L\left( \left\|x^{k+1}-x^k\right\|\right) .
	\end{align*}
	Since $A$ has full column rank, the following inequality is obtained.
	\begin{align*}
		\left\|A^\top\left( \lambda^{k+1}-\lambda^k\right) \right\|\leqslant\sqrt{\left\|AA^\top\right\|}\left\|\lambda^{k+1}-\lambda^k\right\|.
	\end{align*}
	It is similar for $A^\top\left( z^k-z^{k+1}\right) $.
	Finally, we consider the term $\left\|\nabla_{y}f\left( x^{k+1},y^{k+1}\right) -\nabla_{y}f\left( x^{k},y^{k+1}\right) \right\|$. By the definition of $f\left( x,y\right) $, it can be rewritten as
	\begin{align*}
		\left\|\right.
		\left\langle \alpha^{k+1},a^{k+1}\alpha^{k+1}+b^{k+1}\beta^{k+1}+c^{k+1}\gamma^{k+1}-\rho\right\rangle -\left\langle \alpha^{k},a^{k+1}\alpha^{k}+b^{k+1}\beta^{k}+c^{k+1}\gamma^{k}-\rho\right\rangle\\
		\left\langle \beta^{k+1},a^{k+1}\alpha^{k+1}+b^{k+1}\beta^{k+1}+c^{k+1}\gamma^{k+1}-\rho\right\rangle -\left\langle \beta^{k},a^{k+1}\alpha^{k}+b^{k+1}\beta^{k}+c^{k+1}\gamma^{k}-\rho\right\rangle\\
		\left\langle \gamma^{k+1},a^{k+1}\alpha^{k+1}+b^{k+1}\beta^{k+1}+c^{k+1}\gamma^{k+1}-\rho\right\rangle -\left\langle \gamma^{k},a^{k+1}\alpha^{k}+b^{k+1}\beta^{k}+c^{k+1}\gamma^{k}-\rho\right\rangle
		\left.\right\|.
	\end{align*}
	For the first term, we have 
	\begin{align*}
		&\left\|\left\langle \alpha^{k+1},a^{k+1}\alpha^{k+1}+b^{k+1}\beta^{k+1}+c^{k+1}\gamma^{k+1}-\rho\right\rangle -\left\langle \alpha^{k},a^{k+1}\alpha^{k}+b^{k+1}\beta^{k}+c^{k+1}\gamma^{k}-\rho\right\rangle\right\|\notag\\
		\leqslant&|\left\langle \alpha^{k+1},\alpha^{k+1}\right\rangle-\left\langle \alpha^{k},\alpha^{k}\right\rangle|+|\left\langle \alpha^{k+1},\beta^{k+1}\right\rangle-\left\langle \alpha^{k},\beta^{k}\right\rangle|\\
		&+|\left\langle \alpha^{k+1},\gamma^{k+1}\right\rangle-\left\langle \alpha^{k},\gamma^{k}\right\rangle|+\left\langle \rho,\alpha^{k+1}-\alpha^{k}\right\rangle\\
		=&|\left\langle \alpha^{k+1},\alpha^{k+1}\right\rangle-\left\langle \alpha^{k},\alpha^{k+1}\right\rangle+\left\langle \alpha^{k},\alpha^{k+1}\right\rangle-\left\langle \alpha^{k},\alpha^{k}\right\rangle|\\
		&+|\left\langle \alpha^{k+1},\beta^{k+1}\right\rangle-\left\langle \alpha^{k},\beta^{k+1}\right\rangle+\left\langle \alpha^{k},\beta^{k+1}\right\rangle-\left\langle \alpha^{k},\beta^{k}\right\rangle|\\
		&+|\left\langle \alpha^{k+1},\gamma^{k+1}\right\rangle-\left\langle \alpha^{k},\gamma^{k+1}\right\rangle+\left\langle \alpha^{k},\gamma^{k+1}\right\rangle-\left\langle \alpha^{k},\gamma^{k}\right\rangle|+\left\langle \rho,\alpha^{k+1}-\alpha^{k}\right\rangle\\
		\leqslant&\left( \left\|\alpha^{k+1}\right\|+\left\|\alpha^{k}\right\|+\left\|\beta^{k+1}\right\|+\left\|\gamma^{k+1}\right\|+\left\|\rho\right\|\right) \left\|\alpha^{k+1}-\alpha^k\right\|+\left\|\alpha^{k}\right\|\left\|\beta^{k+1}-\beta^k\right\|\\
		&+\left\|\alpha^{k}\right\|\left\|\gamma^{k+1}-\gamma^k\right\|,
	\end{align*}
	where the first inequality follows from the constraint of $y$. Since the sequence $\left\lbrace x^k\right\rbrace$ is bounded, we have 
	\begin{align*}
		\left\|\nabla_{y}f\left( x^{k+1},y^{k+1}\right) -\nabla_{y}f\left( x^{k},y^{k+1}\right) \right\|\leqslant C\left\|x^{k+1}-x^k\right\|,
	\end{align*}
	where $C$ is a constant. Then, the boundness of $\left\|D^{k+1}\right\|$ is obtained easily.
\end{proof}
\begin{theorem}[Global convergence]\label{thmconvergence}
	Let$\left\lbrace \left( y^{k},x^{k},p^{k},z^{k},\lambda^{k}\right) \right\rbrace$ be the sequence generated by LPADMM. It has at least a limit point $\left( y^{*},x^{*},p^{*},z^{*},\lambda^{*}\right) $, and any limit point $\left( y^{*},x^{*},p^{*},z^{*},\lambda^{*}\right) $ is a stationary point. That is, $0\in\partial\mathcal{L}_{\eta}\left( y^{*},x^{*},p^{*},z^{*},\lambda^{*}\right) $. Moreover, since $\mathcal{L}_{\eta}$ is a K{\L} function, $\left\lbrace \left( y^{k},x^{k},p^{k},z^{k},\lambda^{k}\right) \right\rbrace $ converges globally to the unique point $\left( y^{*},x^{*},p^{*},z^{*},\lambda^{*}\right) $.
\end{theorem}
\begin{proof}
	By Proposition \ref{prop:boundedness}, limit points exist. Proposition \ref{prop:subgrad_bound} and the sufficient descent condition (\ref{final}) satisfy the K{\L} framework requirements \cite{wang2019global}, establishing global convergence.
\end{proof}

\section{Iteration Complexity}
We proceed to analyze the iteration complexity. We begin by formally defining the concept of an $\epsilon$-stationary point, followed by establishing the $O(1/\epsilon^2)$ iteration complexity to reach such points. To facilitate our analysis, we first introduce the following composite variable:
\begin{align*}
	w=\begin{pmatrix}
		y\\
		x\\
		p\\
		z
	\end{pmatrix}, \quad \mathcal{L}_{\eta}^{N+1}=\mathcal{L}_{\eta}\left( y^{N+1},x^{N+1},p^{N+1},z^{N+1},\lambda^{N+1}\right) .
\end{align*}
\begin{definition}[$\epsilon$-stationary solution]\label{def:epsilon-stationary} A solution $w^*= \left( y^*, x^*, p^*, z^*\right) $ is said to be an $\epsilon$-stationary solution of problem (\ref{proapp}) if there exists a Lagrange multiplier $\lambda^*$ satisfying:
	\begin{align*}
	\begin{cases}
	\mathrm{dist }\left( -\nabla_{y}f(x^*,y^*),\partial\delta_{\mathcal{Y}}(y^*)\right) \leqslant\epsilon\\
	\mathrm{dist }\left( -\nabla_{x}f(x^*,y^*)-A^\top\lambda^*,\partial\delta_{\mathcal{X}}(x^*)\right) \leqslant\epsilon\\
	\mathrm{dist }\left( -\xi(p^*-z^*),\partial\delta(p^*)\right) \leqslant\epsilon\\
	\left\|\xi(z^*-p^*)-\lambda^*\right\|\leqslant\epsilon\\
	\left\|Ax^*-z^*\right\|\leqslant\epsilon
	\end{cases}
	\end{align*}
\end{definition}
We first establish a key bound for the augmented Lagrangian function.
\begin{lemma}[Lower boundedness of $\mathcal{L}_{\eta}$]\label{Lbound}
	Let the sequence $\left\lbrace \left( y^{k},x^{k},p^{k},z^{k},\lambda^{k}\right) \right\rbrace $ be the sequence generated by LPADMM with parameters satisfying: $\mu_1>0, \mu_2>L, \mu_3>\frac{2\xi^2}{\eta}$ and $\eta>2\xi$. Then, we have
	\begin{align*}
	\mathcal{L}_{\eta}^{N+1}\ge h^*,\quad\forall N\ge0.
	\end{align*}
\end{lemma}
\begin{proof}
From (\ref{aulag}) and the iteration scheme, it can be obtained that
\begin{align*}
&\mathcal{L}_{\eta}^{N+1}=h\left( x^{N+1},y^{N+1},z^{N+1},p^{N+1}\right) +\left\langle \lambda^{N+1},Ax^{N+1}-z^{N+1}\right\rangle+\frac{\eta}{2}\left\|Ax^{N+1}-z^{N+1}\right\|^2\\
=&f\left( x^{N+1},y^{N+1}\right) +\frac{\xi}{2}\left\|p^{N+1}-z^{N+1}\right\|^2+\frac{\xi}{2}\left\|Ax^{N+1}-z^{N+1}\right\|^2-\xi\left\langle p^{N+1}-z^{N+1},Ax^{N+1}-z^{N+1}\right\rangle\\
&+\frac{\eta-\xi}{2}\left\|Ax^{N+1}-z^{N+1}\right\|^2\\
=&f\left( x^{N+1},y^{N+1}\right) +\frac{\xi}{2}\left\|Ax^{N+1}-p^{N+1}\right\|^2+\frac{\eta-\xi}{2}\left\|Ax^{N+1}-z^{N+1}\right\|^2\\
\ge&h^*,
\end{align*}
where $h^*$ denotes the optimal value of $h\left( x,y,z,p\right) $ without linear constraint.
\end{proof}

\begin{theorem}[Iteration complexity]
	Suppose that the sequence $\left\lbrace \left( y^{k},x^{k},p^{k},z^{k},\lambda^{k}\right) \right\rbrace $ is generated by LPADMM. $\mu_1>0, \mu_2>L, \mu_3>\frac{2\xi^2}{\eta}$ and $\eta>2\xi$. Then, we can find an $\epsilon$-stationary solution of problem (\ref{proapp}), where $\left( y^{k},x^{k},p^{k},z^{k},\lambda^{k}\right) $ be the first iteration that satisfies
	\begin{align*}
		\Delta_{\hat{k}}:=\left\|w^{\hat{k}}-w^{\hat{k}-1}\right\|^2\leqslant\epsilon^2/\max\left\lbrace \omega_1,\omega_2,\omega_3,\omega_4,\omega_5\right\rbrace ,
	\end{align*}
where $\omega_1:=\frac{2\xi^2}{\eta^2}$, $\omega_2:=2\xi^2$, $\omega_3=\max\left\lbrace 2\xi^2,2\mu_3^2\right\rbrace$, $\omega_4:=\max\left\lbrace 4\left( L^2+\mu_2^2\right) ,16\xi^2+8\eta^2\right\rbrace$, $\omega_5:=\max\left\lbrace2C^2,2\mu_1^2 \right\rbrace$.\\
	Moreover, $\hat{k}$ is no more than 
	\begin{align*}
		T:=\left \lceil\frac{\max\left\lbrace \omega_1,\omega_2,\omega_3,\omega_4,\omega_5\right\rbrace} {\nu\epsilon^2}\left( \mathcal{L}_{\eta}^1-h^*\right)   \right \rceil,
	\end{align*}
where $\nu:=\min\left\lbrace \left( \eta+\mu_2-L\right) ,\left( \frac{\eta}{2}-\frac{2\xi^2}{\eta}\right) ,\left( \mu_3-\frac{2\xi^2}{\eta}\right) ,\mu_1\right\rbrace $.
\end{theorem}
\begin{proof}
By summing (\ref{final}) over $k=1,2,\cdots,N$, we get
\begin{align*}
	\mathcal{L}_{\eta}^1-\mathcal{L}_{\eta}^{N+1}\ge\nu\sum_{k=1}^{N}\left\|w^{k+1}-w^k\right\|^2.
\end{align*}
Then, 
\begin{align*}
	\min_{2\leqslant k\leqslant N+1}\Delta_{k}&\leqslant\frac{1}{T}\sum_{k=2}^{N+1}\Delta_{k}\\
	&\leqslant\frac{1}{\nu T}\left( \mathcal{L}_{\eta}^1-h^*\right) .
\end{align*}
Considering the $\epsilon$-stationary point, the bound of $\partial\mathcal{L}_{\eta}^{N+1}$ needs to be calculated. Using the results in Proposition \ref{prop:subgrad_bound}, we have
\begin{align}\label{new1}
\left\|Ax^{N+1}-z^{N+1}\right\|^2\leqslant\frac{2\xi^2}{\eta^2}\Delta_{N+1}.
\end{align}
\begin{align}
	\mathrm{dist}\left( -\xi\left( p^{N+1}-z^{N+1}\right) ,\partial\delta_{\mathcal{Z}}\left( p^{N+1}\right) \right) &\leqslant2\xi^2\left( \left\|z^{N+1}-z^N\right\|^2\right) +2\mu_3^2\left( \left\|p^{N+1}-p^N\right\|^2\right) \notag\\
	&\leqslant\sqrt{\omega_3 \Delta_{N+1}}
\end{align}
\begin{align}
&\mathrm{dist }\left( -\nabla_{x}f\left( x^{N+1},y^{N+1}\right) -A^\top\lambda^{N+1},\partial\delta_{\mathcal{X}}\left( x^{N+1}\right) \right) \notag\\
\leqslant&4\left[ \left( L^2+\mu_2^2\right) \left\|x^{N+1}-x^N\right\|^2+2\left\|\lambda^{N+1}-\lambda^N\right\|^2+2\eta^2\left\|z^{N+1}-z^N\right\|^2\right]\notag \\
\leqslant&\sqrt{\omega_4 \Delta_{N+1}}
\end{align}
\begin{align}\label{new4}
	\mathrm{dist }\left( -\nabla_{y}f\left( x^{N+1},y^{N+1}\right) ,\partial\delta_{\mathcal{Y}}\left( y^{N+1}\right) \right) &\leqslant2C^2\left\|x^{N+1}-x^N\right\|^2+2\mu_1^2\left\|y^{N+1}-y^N\right\|^2\notag\\
	&\leqslant\sqrt{\omega_5 \Delta_{N+1}}
\end{align}
Combining (\ref{new1})-(\ref{new4}), we can get the upper bound of iteration complexity.
\end{proof}

	\bibliographystyle{unsrt}
	\bibliography{quantum}

\begin{thebibliography}{10}

\bibitem{einstein1935can}
A.~Einstein, B.~Podolsky, and N.~Rosen.
\newblock Can quantum-mechanical description of physical reality be considered
  complete?
\newblock {\em Physical Review}, 47(10):777, 1935.

\bibitem{shor1994algorithms}
P.~W. Shor.
\newblock Algorithms for quantum computation: discrete logarithms and
  factoring.
\newblock In {\em Proceedings 35th Annual Symposium on Foundations of Computer
  Science}, pages 124--134. IEEE, 1994.

\bibitem{ekert1991quantum}
A.~K. Ekert.
\newblock Quantum cryptography based on {B}ell's theorem.
\newblock {\em Physical Review Letters}, 67(6):661, 1991.

\bibitem{georgescu2014quantum}
I.~M. Georgescu, S.~Ashhab, and F.~Nori.
\newblock Quantum simulation.
\newblock {\em Reviews of Modern Physics}, 86(1):153--185, 2014.

\bibitem{peres1996separability}
A.~Peres.
\newblock Separability criterion for density matrices.
\newblock {\em Physical Review Letters}, 77(8):1413, 1996.

\bibitem{horodecki1996necessary}
M.~Horodecki, P.~Horodecki, and R.~Horodecki.
\newblock On the necessary and sufficient conditions for separability of mixed
  quantum states.
\newblock {\em Physics Letters A}, 223(1), 1996.

\bibitem{eggeling2001separability}
Tilo Eggeling and Reinhard~F Werner.
\newblock Separability properties of tripartite states with ${U}\times
  {U}\times {U}$ symmetry.
\newblock {\em Physical Review A}, 63(4):042111, 2001.

\bibitem{han2019construction}
Kyung~Hoon Han and Seung-Hyeok Kye.
\newblock Construction of three-qubit biseparable states distinguishing kinds
  of entanglement in a partial separability classification.
\newblock {\em Physical Review A}, 99(3):032304, 2019.

\bibitem{ha2016construction}
Kil-Chan Ha and Seung-Hyeok Kye.
\newblock Construction of three-qubit genuine entanglement with bipartite
  positive partial transposes.
\newblock {\em Physical Review A}, 93(3):032315, 2016.

\bibitem{chen2019separability}
L.~Chen, D.~Chu, L.~Qian, et~al.
\newblock Separability of completely symmetric states in a multipartite system.
\newblock {\em Physical Review A}, 99(3):032312, 2019.

\bibitem{choudhary2025unconditionally}
Swati Choudhary, Ujjwal Sen, and Saronath Halder.
\newblock Unconditionally superposition-robust entangled state in all
  multiparty quantum systems.
\newblock {\em Physical Review A}, 112(1):012404, 2025.

\bibitem{horodecki2009quantum}
Ryszard Horodecki, Pawe{\l} Horodecki, Micha{\l} Horodecki, and Karol
  Horodecki.
\newblock Quantum entanglement.
\newblock {\em Reviews of Modern Physics}, 81(2):865--942, 2009.

\bibitem{gharibian2008strong}
Sevag Gharibian.
\newblock Strong {NP}-hardness of the quantum separability problem.
\newblock {\em arXiv preprint arXiv:0810.4507}, 2008.

\bibitem{han2013successive}
D.~Han and L.~Qi.
\newblock A successive approximation method for quantum separability.
\newblock {\em Frontiers of Mathematics in China}, 8(6):1275--1293, 2013.

\bibitem{Lancien_2015}
Cécilia Lancien, Otfried Gühne, Ritabrata Sengupta, and Marcus Huber.
\newblock Relaxations of separability in multipartite systems: Semidefinite
  programs, witnesses and volumes.
\newblock {\em Journal of Physics A: Mathematical and Theoretical},
  48(50):505302, nov 2015.

\bibitem{li2020separability}
Y.~Li and G.~Ni.
\newblock Separability discrimination and decomposition of m-partite quantum
  mixed states.
\newblock {\em Physical Review A}, 102(1):012402, 2020.

\bibitem{lin2022low}
M.~M. Lin and M.~T. Chu.
\newblock Low-rank approximation to entangled multipartite quantum systems.
\newblock {\em Quantum Information Processing}, 21(4):120, 2022.

\bibitem{chen2021entanglement}
Changbo Chen, Changliang Ren, Hongqing Lin, and He~Lu.
\newblock Entanglement structure detection via machine learning.
\newblock {\em Quantum Science and Technology}, 6(3):035017, 2021.

\bibitem{asif2023entanglement}
Naema Asif, Uman Khalid, Awais Khan, Trung~Q Duong, and Hyundong Shin.
\newblock Entanglement detection with artificial neural networks.
\newblock {\em Scientific Reports}, 13(1):1562, 2023.

\bibitem{yang2015alternating}
L.~Yang, T.~K. Pong, and X.~Chen.
\newblock Alternating direction method of multipliers for a class of nonconvex
  and nonsmooth problems with applications to background/foreground extraction.
\newblock {\em SIAM Journal on Imaging Sciences}, 10(1):74--110, 2017.

\bibitem{he2024admm}
Ang He, Heng Pan, Yueyue Dai, Xueming Si, Chau Yuen, and Yan Zhang.
\newblock {ADMM} for mobile edge intelligence: A survey.
\newblock {\em IEEE Communications Surveys \& Tutorials}, 2024.

\bibitem{shen2014augmented}
Yuan Shen, Zaiwen Wen, and Yin Zhang.
\newblock Augmented {L}agrangian alternating direction method for matrix
  separation based on low-rank factorization.
\newblock {\em Optimization Methods and Software}, 29(2):239--263, 2014.

\bibitem{sun2014alternating}
Dennis~L Sun and Cedric Fevotte.
\newblock Alternating direction method of multipliers for non-negative matrix
  factorization with the beta-divergence.
\newblock In {\em 2014 IEEE International Conference on Acoustics, Speech and
  Signal Processing (ICASSP)}, pages 6201--6205. IEEE, 2014.

\bibitem{xu2012alternating}
Yangyang Xu, Wotao Yin, Zaiwen Wen, and Yin Zhang.
\newblock An alternating direction algorithm for matrix completion with
  nonnegative factors.
\newblock {\em Frontiers of Mathematics in China}, 7(2):365--384, 2012.

\bibitem{gao2024low}
Kaixin Gao, Zheng-Hai Huang, and Lulu Guo.
\newblock Low-rank matrix recovery problem minimizing a new ratio of two norms
  approximating the rank function then using an {ADMM}-type solver with
  applications.
\newblock {\em Journal of Computational and Applied Mathematics}, 438:115564,
  2024.

\bibitem{sun2016deep}
Jian Sun, Huibin Li, Zongben Xu, et~al.
\newblock Deep {ADMM}-net for compressive sensing mri.
\newblock {\em Advances in Neural Information Processing Systems}, 29, 2016.

\bibitem{attouch2010proximal}
H.~Attouch, J.~Bolte, P.~Redont, et~al.
\newblock Proximal alternating minimization and projection methods for
  nonconvex problems: An approach based on the {K}urdyka-{\l}ojasiewicz
  inequality.
\newblock {\em Mathematics of Operations Research}, 35(2):438--457, 2010.

\bibitem{deng2016global}
W.~Deng and W.~Yin.
\newblock On the global and linear convergence of the generalized alternating
  direction method of multipliers.
\newblock {\em Journal of Scientific Computing}, 66(3):889--916, 2016.

\bibitem{hong2016convergence}
Mingyi Hong, Zhi-Quan Luo, and Meisam Razaviyayn.
\newblock Convergence analysis of alternating direction method of multipliers
  for a family of nonconvex problems.
\newblock {\em SIAM Journal on Optimization}, 26(1):337--364, 2016.

\bibitem{guo2017convergence}
K.~Guo, D.~Han, and T.~Wu.
\newblock Convergence of alternating direction method for minimizing sum of two
  nonconvex functions with linear constraints.
\newblock {\em International Journal of Computer Mathematics},
  94(8):1653--1669, 2017.

\bibitem{wang2019global}
Y.~Wang, W.~Yin, and J.~Zeng.
\newblock Global convergence of {ADMM} in nonconvex nonsmooth optimization.
\newblock {\em Journal of Scientific Computing}, 78:29--63, 2019.

\bibitem{jungnitsch2011taming}
B.~Jungnitsch, T.~Moroder, and O.~G{\"u}hne.
\newblock Taming multiparticle entanglement.
\newblock {\em Physical Review Letters}, 106(19):190502, 2011.

\bibitem{nesterov2018lectures}
Y.~Nesterov.
\newblock {\em Lectures on convex optimization}, volume 137.
\newblock Springer, 2018.

\bibitem{nocedal1999numerical}
J.~Nocedal and S.~J. Wright.
\newblock {\em Numerical optimization}.
\newblock Springer, 1999.

\end{thebibliography}
\end{document}